\documentclass[sn-mathphys,Numbered, a4paper]{sn-jnl}

\usepackage{graphicx}%
\usepackage{multirow}%
\usepackage{amsmath,amssymb,amsfonts}%
\usepackage{amsthm}%
\usepackage{mathrsfs}%
\usepackage[title]{appendix}%
\usepackage{xcolor}%
\usepackage{textcomp}%
\usepackage{manyfoot}%
\usepackage{booktabs}%
\usepackage{algorithm}%
\usepackage{algorithmicx}%
\usepackage{algpseudocode}%
\usepackage{listings}%
\usepackage{tikz}
\usetikzlibrary{calc, angles, quotes, arrows.meta}
\usepackage{adjustbox}

\theoremstyle{thmstyleone}%
\newtheorem{theorem}{Theorem}
\newtheorem{proposition}[theorem]{Proposition}%
\newtheorem{lemma}[theorem]{Lemma}
\newtheorem{corollary}[theorem]{Corollary}

\theoremstyle{thmstyletwo}%
\newtheorem{remark}{Remark}%

\theoremstyle{thmstylethree}%
\newtheorem{definition}{Definition}%

\DeclareMathOperator{\curv}{curv}

\DeclareMathOperator{\vol}{vol}

\DeclareMathOperator{\ev}{ev}
\DeclareMathOperator{\Poly}{Poly}
\DeclareMathOperator{\SU}{SU}
\DeclareMathOperator{\SO}{SO}
\DeclareMathOperator{\orr}{or}

\DeclareMathOperator{\Hol}{Hol}
\DeclareMathOperator{\Rre}{Re}
\newcommand{\z}{\mathrm{z}}
\newcommand{\tildegamma}{{\tilde{\gamma}}}
\newcommand{\tildesigma}{{\tilde{\sigma}}}

\begin{document}

\title{Coherent loop states and angular momentum}


\author[1,2]{\fnm{Bruce} \sur{Bartlett}}

\author[1,3]{\fnm{Nzaganya} \sur{Nzaganya}}

\affil[1]{\orgdiv{Mathematics Division}, \orgname{Stellenbosch University}}


\affil[2]{\orgname{National Institute of Theoretical and Computational Sciences (NITHECS)}, \country{South Africa}}

\affil[3]{\orgdiv{Department of Mathematics, Mkwawa University College of Education}, \orgname{University of Dar es Salaam}, \country{Tanzania}}


\abstract{We study Bohr-Sommerfeld states in the context of the irreducible representations of SU(2). These states offer a precise bridge between the classical and quantum descriptions of angular momentum. We show that they recover the usual basis of angular momentum eigenstates used in physics, and give a self-contained proof in this setting of the formula of Bothwick, Paul and Uribe for the asymptotics of the inner product of arbitrary coherent loop states. As an application, we use these states to derive Littlejohn and Yu's geometric formula for the asymptotics of the Wigner matrix elements.
}




\maketitle

\section{Introduction}
The purpose of this paper is to study the {\em Bohr-Sommerfeld states} \cite{borthwick1995legendrian} (which we will call {\em coherent loop states} in our setting) in the context of the irreducible representations of $\SU(2)$, and to use these states to derive the `spherical area' formula stated in \cite{littlejohn2009uniform} for the asymptotics of the matrix elements of these representations. We will see that the general theory in \cite{borthwick1995legendrian} takes a particularly simple and elegant form in this context, where the geometry of the Hopf fibration $S^3 \rightarrow S^2$ will play a central role. 

From the viewpoint of physics, the key feature of coherent loop states is that they allow one to actually make rigorous many of the intuitive classical mental images we have for spin angular momentum (such as in \cite{Ponzano1968, brussaard1957classical, biedenharn1981racah}), since they offer a precise and convenient bridge from the classical to the quantum world.

\subsection*{Borthwick, Paul and Uribe's asymptotic formula}

Recall that in geometric quantization of K\"{a}hler manifolds, one starts with a compact holomorphic manifold $M$ and a Hermitian line bundle $L$ over $M$ which is positive in the sense that its  associated connection $\nabla$ induces a symplectic form $\omega = i \text{curv}(\nabla)$ and Riemannian metric on $M$. We think of $M$ as a classical phase space, and we associate to it the quantum Hilbert space
\[
 V_k = H^0(M, L^k)
\]
of holomorphic sections of $L^k$ (note that $V_k$ is finite-dimensional since $M$ is compact). 

Let $\Lambda \subset M$ be a Lagrangian submanifold of $M$ (i.e. $\dim \Lambda = \frac{1}{2} \dim M$ and $\omega|_\Lambda = 0$) which is {\em Bohr-Sommerfeld of order $k$} in the sense that the holonomy of $L^k$ restricted to $\Lambda$ is trivial. In \cite{borthwick1995legendrian}, Borthwick, Paul and Uribe showed how to construct a vector (the `Bohr-Sommerfeld state')
\[
 \Psi^{(k)}_{\tilde{\Lambda}} \in V_k
\]
associated with a parallel-transported lift $\tilde{\Lambda}$ of $\Lambda$ to the unit circle bundle $P \subset L^\vee$, and computed the leading asymptotics as $k\rightarrow \infty$ of inner products 
\[
 \left\langle \Psi_{\tilde{\Lambda}}^{(k)}, \Psi_{\tilde{\Lambda'}}^{(k)} \right\rangle
\]
of such states as an integral over $\Lambda \cap \Lambda'$. 

In the special case where $M$ is a Riemann surface, $\Lambda$ and $\Lambda'$ are simply loops $\gamma$ and $\sigma$ in $M$ where the holonomy of $L^k$ vanishes, and their formula reads as follows in the notation which we introduce in this paper.

\begin{theorem}[\cite{borthwick1995legendrian}] \label{BPUTheorem}Let $L$ be a positive holomorphic Hermitian line bundle over a compact Riemann surface $M$. Let $\gamma$ and $\sigma$ be Bohr-Sommerfeld loops in $M$, parameterized by arclength and intersecting transversely, with parallel-transported lifts $\tildegamma$ and $\tildesigma$ in the unit circle bundle $P \subset L^\vee$ respectively. Then as $k\rightarrow \infty$ through joint Bohr-Sommerfeld values for $\gamma$ and $\sigma$,
\[
 \left\langle \Psi_{\tildegamma}^{(k)}, \Psi_{\tildesigma}^{(k)} \right\rangle \sim \sqrt{2} \sum_{x \in \gamma \cap \sigma} \frac{\omega_x^k e^{i \orr_x (\theta_x /2 - \pi/4)}}{\sqrt{\sin \theta_x}}
\]
where $\theta_x \in (0, \pi)$ is the angle between $\gamma$ and $\sigma$ at $x \in M$, $\orr_x = \pm 1$ is the orientation of that angle (it equals $+1$ if rotating $\gamma'$ to $\sigma'$ agrees with the orientation of $M$, and $-1$ otherwise), and $\omega_x = \tildegamma_x / \tildesigma_x \in U(1)$.
\end{theorem}

In \cite{borthwick1995legendrian} this result was applied to study modular forms. In this case, $M$ is the quotient of the upper half plane $\mathbb{H}$ by a Fuchsian group of the first kind, $L$ is the holomorphic tangent bundle, and $P \cong \mathbb{H} \times S^1$ is the unit circle bundle of the upper half plane. The coherent loop states $\Psi^{(k)}_{\tildegamma}$ turn out to be well-known Poincar\'{e} series.

In this paper, we will study instead the case relevant to angular momentum. In our case, $M$ is $S^2 \cong \mathbb{CP}^1$, $L$ is the dual of the tautological line bundle, and $P=S^3$ is the Hopf fibration over $S^2$. The coherent loop states $\Psi^{(k)}_\tildegamma$ turn out to be the well-known angular momentum eigenstates $|jm\rangle$ familiar to physicists.

In \cite{borthwick1995legendrian}, Theorem \ref{BPUTheorem} is proved as a corollary of a more general result which holds for all dimensions of $M$. The proof of this more general result relies heavily on analysis, in particular the machinery of Fourier integral operators of Hermite type developed by De Monvel and Guillemin in \cite{de1981spectral}. In contrast, in Theorem \ref{mainthm1} (our first main result), we will give an elementary and mostly self-contained proof of Theorem \ref{BPUTheorem} in our setting, which only requires as input the the complex version (due to Pemantle and Wilson \cite{pemantle2010asymptotic}) of the well-known stationary phase principle.

\subsection*{Angular momentum eigenstates as coherent loop states}

It is known that the Bohr-Sommerfeld loops $\gamma_m$ on $S^2$ of order $k = 2j$, having constant height $\z$, are precisely those occurring at the discrete heights
\[
 \z_m = m/j, \quad m \in \{-j, -j+1, \dots, j\}\, .
\]
Indeed, this goes back to the roots of quantum mechanics, the old quantum theory \cite{sommerfeld, wilson}. The intuition behind the classical to quantum correspondence is that such a loop is to be thought of as the directions in space where the spin vector of the quantum state $|jm\rangle$ are localized, for large $j$.

\begin{figure}
	\centering
	\raisebox{-0.5\height}{\includegraphics[width=0.3\textwidth]{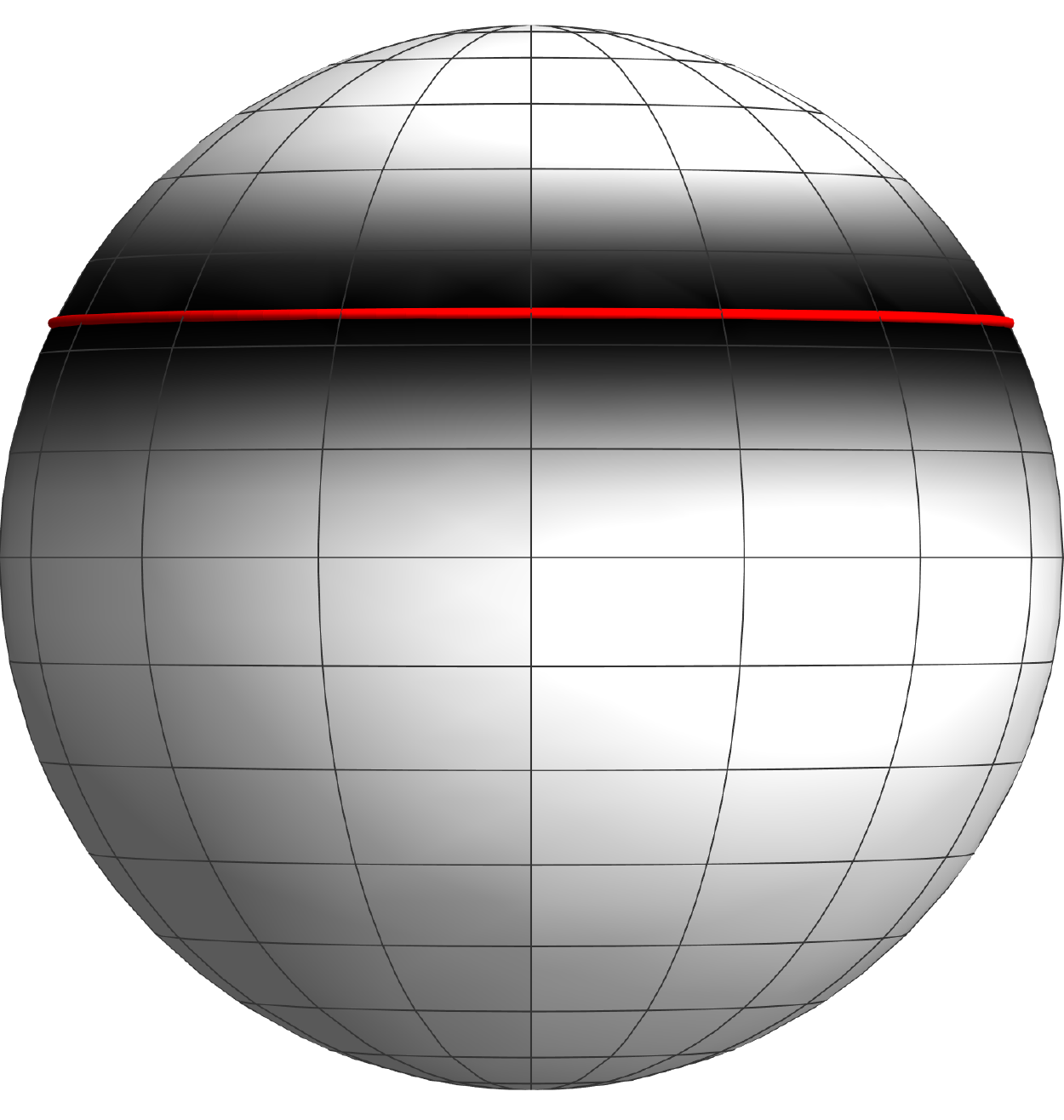}}
	\raisebox{-0.5\height}{\includegraphics[width=0.03\textwidth]{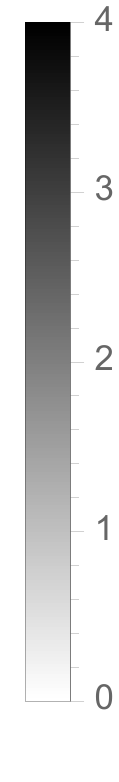}}

	\caption{\label{local_state} A plot of the fibrewise norm $|\Psi^{k}_{\tildegamma_m}(x)|$ for $x \in S^2$ when $k=50$ and $m=11$. This corresponds to the angular momentum eigenstate $|j=25, m=11 \rangle$. }
\end{figure}

The point about coherent loop states is that they make this intuition rigorous, even for finite $j$. Each loop $\gamma_m$ in $S^2$ has an associated lift $\tildegamma_m$ to $S^3$, and the associated coherent loop state is indeed a multiple of $|jm\rangle$, as we will see in Section \ref{loopsareangs}:
\[
  \Psi^{(k)}_{\tildegamma_m} \,\, \propto \,\, |jm\rangle .
\]
Now, by definition, $\Psi^{(k)}_{\tildegamma_m}$ is an integral over the loop $\gamma(t)$ in $S^2$, parameterized by arclength, of the {\em coherent state} $\psi_{\tildegamma_m}^{(k)} \in V_k$ associated to the parallel-transported lift $\tildegamma(t) \in S^3$:
\begin{equation} \label{intdefn}
 \Psi^{(k)}_{\tildegamma_m} := \int_0^T \psi_{\tildegamma(t)} dt.
\end{equation}
The support of a coherent state $\psi_p^{(k)}$ where $p \in S^3$ is Gaussian-shaped and localizes around its basepoint $x = \pi(p) \in S^2$ as $k \rightarrow \infty$. Thus the support of a coherent loop state $\Psi^{(k)}_\tildegamma$ localizes on the loop $\gamma \subset S^2$. See Figure \ref{local_state} . 

Indeed, by their definition \eqref{intdefn} as integrals of coherent states over the loop, we see at a glance why $\Psi^{(k)}_{\tildegamma_m}$ is an eigenstate of the rotation operator $U_z(\Delta \phi)$: such a rotation simply advances the loop parameter by $\Delta \phi$, and a phase factor of $e^{i m \Delta \phi}$ is picked up to account for parallel transport. We also see why $\gamma_m$ can only occur at a quantized height: it is to ensure that the coherent states $\psi^{(k)}_{\tildegamma(t)}$ form a loop of vectors in $V_k$ (otherwise this argument would break down). Thus coherent loop states indeed build an elegant bridge between the classical and quantum worlds. 

\subsection*{An exchange of integrals}
From the viewpoint of asymptotics, the key technical contribution that coherent loop states provide is the following. As we have discussed, a common problem in geometric quantization is that two sequences of holomorphic sections 
\[
  s_k, s'_k \in H^0 (M, L^k)
\] 
are given and the task is to derive an asymptotic formula for their global inner product
\begin{equation} \label{globalip}
 \langle s_k, s'_k \rangle := \int_{x \in M} \left(s_k(x), s'_k(x)\right)_x \vol_x
\end{equation}
where $\vol_x = \omega^n / n!$ is the Liouville form of $\omega$, and $\left(s_k(x), s'_k(x)\right)_x$ is the fibrewise inner product in $L_x^k$. To proceed, express the integrand in exponential form:
\[
 \left(s_k(x), s'_k(x)\right)_x = e^{ikS(x)}.
\]
The complex stationary phase principle tells us that the main contributions to the integral are from the critical points $x \in M$ where $dS(x) = 0$. This approach was used for example to great success in \cite{roberts1999classical} to compute the asymptotics of the classical 6j symbols, which served as inspiration for this paper.

Now, suppose that the sections $s_k, s'_k$ are {\em coherent loop states} associated to parallel-transported lifts $\tildegamma$ and $\tildesigma$ in $P$ of Bohr-Sommerfeld loops $\gamma$ and $\sigma$ in $M$:
\[
 s_k = \Psi^{(k)}_\tildegamma, \quad s'_k = \Psi^{(k)}_\tildesigma .
\]
As we saw in \eqref{intdefn}, this means that each state is {\em in itself} an integral, namely of the coherent state attached to each point of the parallel-transported lift in $P$. Therefore, we can {\em exchange the order of integration} in \eqref{globalip} and rewrite the inner product $\langle \Psi_\tildegamma^{(k)}, \Psi_\tildesigma^{(k)} \rangle$ of the coherent loop states as an integral of coherent states over the {\em torus} $S^1 \times S^1$:
\begin{equation} \label{exchangeintegrals}
\int_{x \in S^2} \left( \Psi_\tildegamma (x), \, \Psi_\tildesigma(x) \right)_x \, \vol_x = \int_{s=0}^S \int_{t=0}^T \left\langle \psi_{\tildegamma(s)}^{(k)}, \psi_{\tildesigma(t)}^{(k)} \right\rangle \, ds dt.
\end{equation}
\begin{figure}[t!]
       \centering
    \begin{minipage}[c]{0.388\textwidth}
        \centering
        \raisebox{-0.5\height}{\includegraphics[width=0.9\textwidth]{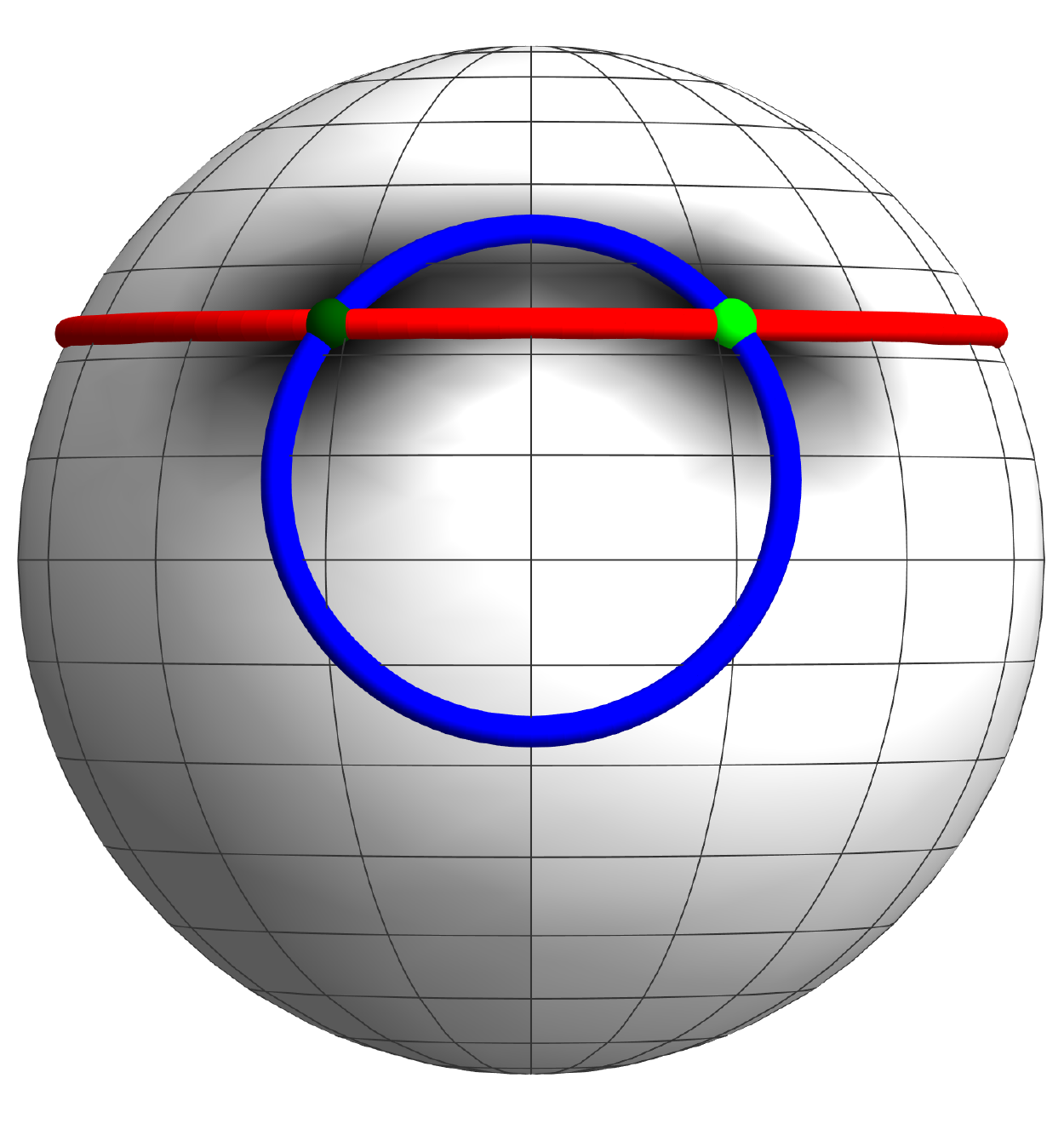}}%
        \raisebox{-0.5\height}{\includegraphics[width=0.1\textwidth]{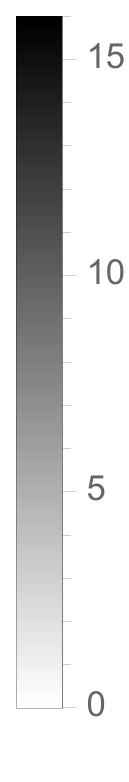}}
        \label{fig:1c}
    \end{minipage}
    \hfill
    \begin{minipage}[c]{0.45\textwidth}
        \centering
        \raisebox{-0.5\height}{\includegraphics[width=0.9\textwidth]{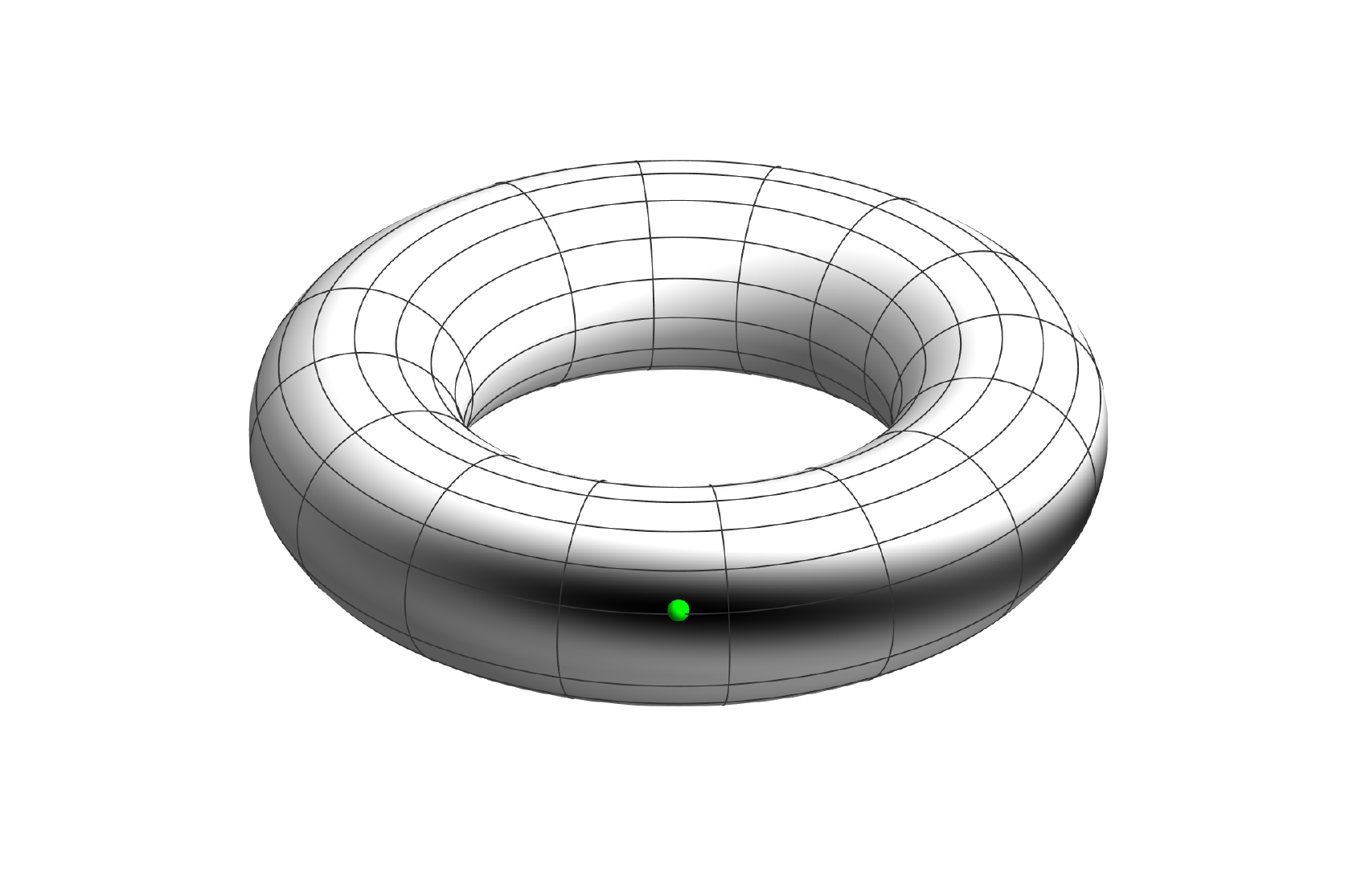}}%
        \raisebox{-0.5\height}{\includegraphics[width=0.1\textwidth]{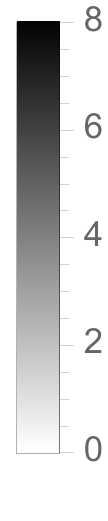}}
        \label{fig:1d}
    \end{minipage}
    
    \vskip\baselineskip
    \centering
    
    \begin{minipage}[c]{0.35\textwidth}
        \centering
        \includegraphics[width=\textwidth]{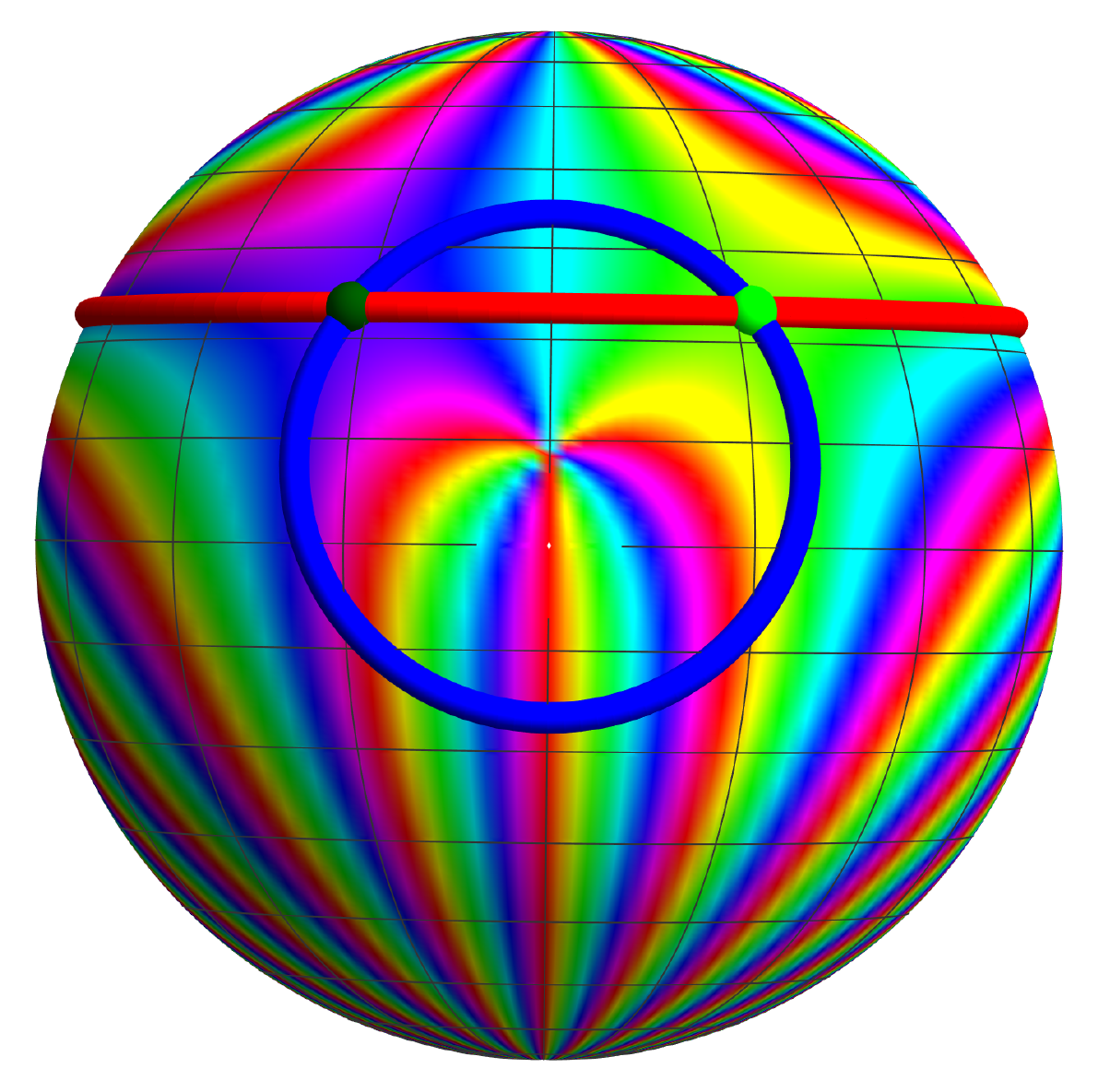}
        \label{fig:1a}
    \end{minipage}
    \begin{minipage}[c]{0.2\textwidth}
        \centering
        \includegraphics[width=\textwidth]{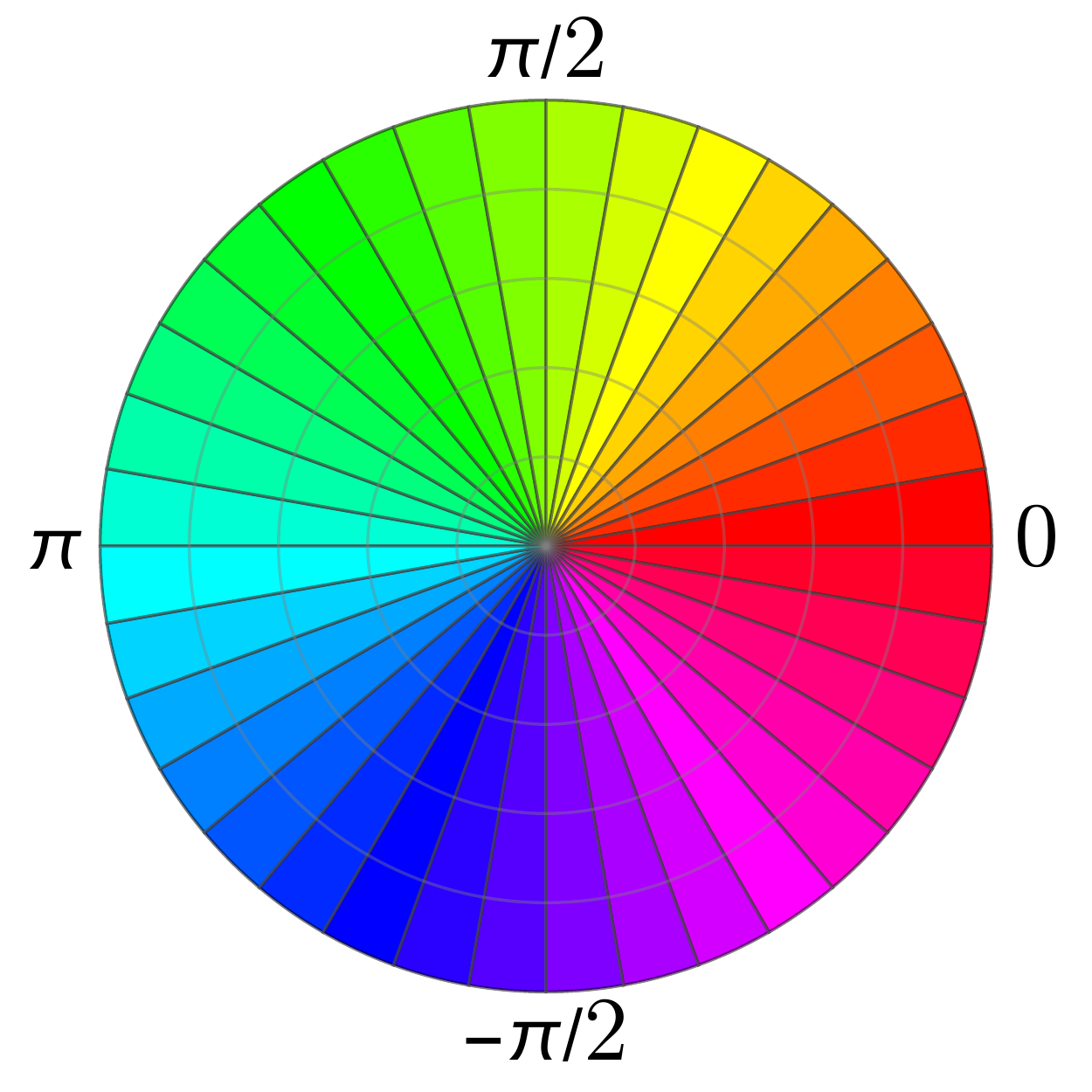}
    \end{minipage}
    \begin{minipage}[c]{0.405\textwidth}
        \centering
        \includegraphics[width=\textwidth]{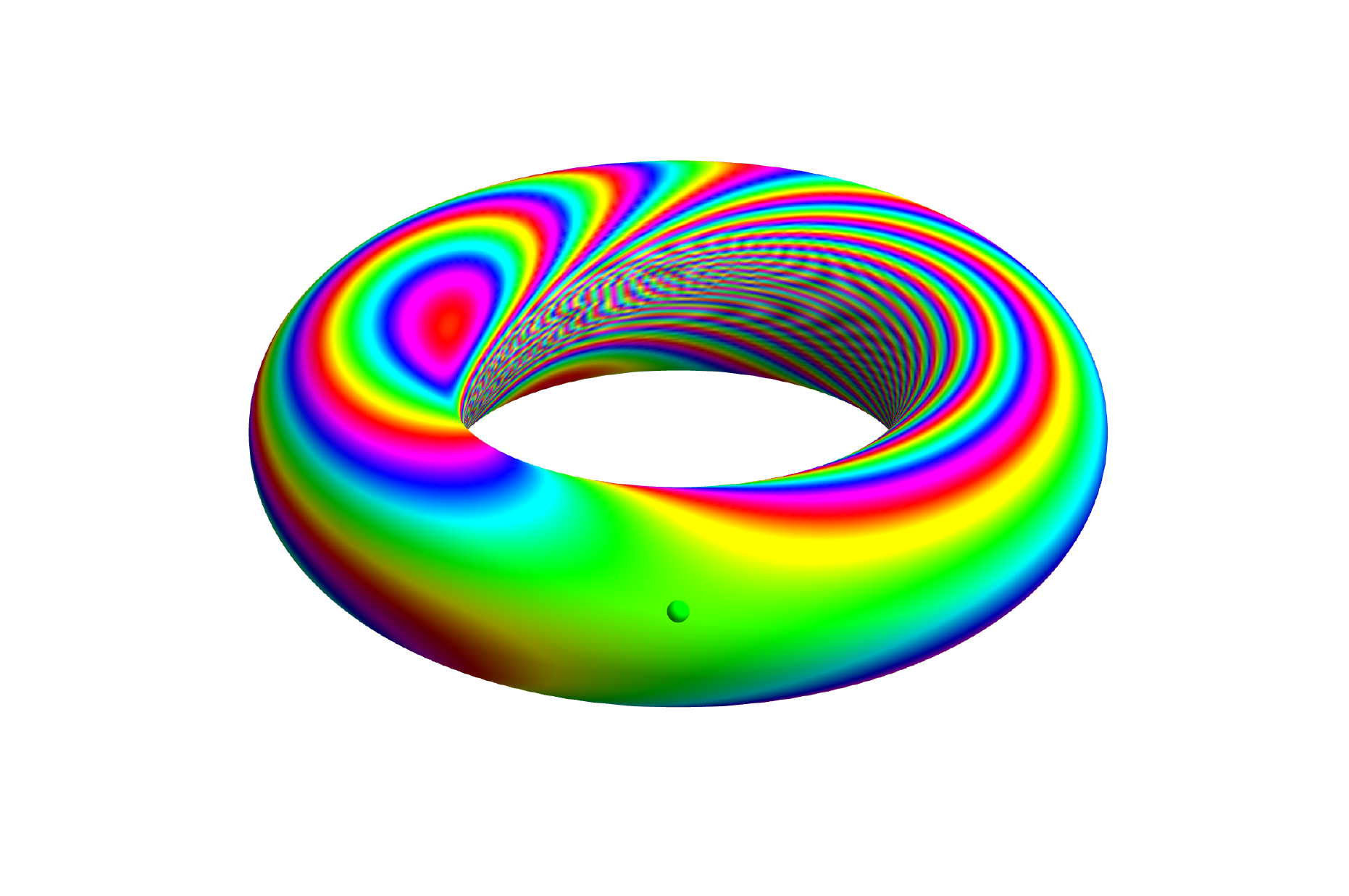}
        \label{fig:1b}
    \end{minipage}
      \caption{\label{torusvssphere} The inner product of coherent loop states  
      for $k=50$ associated to the loops $\gamma$ at height $11/25$ and $\sigma = R_y(\beta)(\rho)$ where $\rho$ has height $22/25$, and $\beta = 1.4$. The left hand plots show the magnitude and phase of $\left( \Psi_\tildegamma (x), \, \Psi_\tildesigma(x) \right)_x$ on $S^2$, while the right hand plots show the magnitude and phase of $\left\langle \psi_{\tildegamma(s)}^{(k)}, \Psi_{\tildesigma(t)}^{(k)} \right\rangle$ on $S^1 \times S^1$. Note that in both cases the critical points correspond to the intersection points of the loops on $S^2$.}
\end{figure}
This is illustrated in Figure \ref{torusvssphere}. The advantage gained is that as an integral over the torus, the integrand $\langle \psi^{(k)}_{\tildegamma(s)}, \psi^{(k)}_{\tildesigma(t)} \rangle$ at $(s,t)$ is a global inner product over $M$ of coherent states localized at $\gamma(t)$ and $\sigma(s)$ in $M$. As we will explain, this is precisely the Bergman kernel $B(x,y)$ on $M \times M$ evaluated at $x = \gamma(s), y=\sigma(t)$. The Bergman kernel is a canonical geometric object whose local behaviour and asymptotics are well understood. Thus the integral on the right hand side of \eqref{exchangeintegrals} is easier to understand and approximate than the integral on the left.

\subsection*{Littlejohn and Yu's spherical area formula}

The {\em Wigner small $d$-matrix} $d^j_{m'm}(\beta)$ is the matrix\footnote{Note that our convention is to use positively oriented, active rotations, so that our $d^j_{m'm}(\beta)$ is equal to Wigner's $d^j_{m'm}(-\beta)$ from \cite{wigner1959}.}of the action of the `rotate counterclockwise by $\beta$ about the positive $y$-axis' operator
\[
 U_y(\beta) = \begin{pmatrix} \cos \beta/2 & \sin \beta/2 \\ - \sin \beta/2 & \cos \beta/2 \end{pmatrix} \in \SU(2)
\]
in the irreducible representation $V_{2j}$ of $SU(2)$ of dimension $2j+1$. Since the rotation operator $U_z(\beta)$ about the $\z$-axis acts diagonally on the standard basis $|j m \rangle$, knowledge of the smal $d$-matrix elements
\[
 d^j_{m'm}(\beta) = \langle jm' | U_y(\beta) | jm \rangle
\]
suffices to compute the matrix elements of arbitrary rotation operators (by expressing such a rotation in terms of $zyz$ Euler angles). The asymptotics of $d^j_{m'm}(\beta)$ for large $j$ has been well studied historically \cite{brussaard1957classical, Ponzano1968, braun1996semiclassics, sokolovski1999semiclassical}, but mainly by using the WKB method to approximate the solution to the Schr\"{o}dinger-style equation in $\beta$ which $d^j_{m'm}(\beta)$ satisfies. Thus, there was something of a disconnect between the geometry of the classical system and the approximation scheme used in the quantum system. 

In \cite{littlejohn2009uniform}, Littlejohn and Yu expressed the asymptotics of $d^j_{m'm}(\beta)$ in a manner intrinsically related to the classical geometry, namely as a cosine function whose frequency is governed by the spherical area of the lunar region enclosed by the loops $\gamma_{m'}$ and $\sigma_m = R_y(\beta) \gamma_{m}$ on $S^2$. However, \cite{littlejohn2009uniform} is written from a physics perspective and they did not give a detailed proof of their formula. Using coherent states, we will give a rigorous proof of a (slightly different) version of their formula in Theorem \ref{lyutheorem}. This is our second main result in this paper. In fact, we first prove a similar `spherical area' cosine formula (Corollary \ref{corsymcurves}) for the inner product of {\em any} two coherent loop states (not just those having constant height), as long as the loops intersect twice transversely at equal angles. From this, we deduce Theorem \ref{lyutheorem} as a special case.

\subsection*{Outline of paper}
In Section \ref{Bergkersec}, we review coherent states, the Bergman kernel, and the construction of coherent loop states in a general setting. We describe how a group of symmetries will act on such states, and we review the complex stationary phase formula of Pemantle and Wilson. In Section \ref{cohstatesS2} we discuss coherent states on $S^2$ in terms of the geometry of the Hopf fibration $S^3 \rightarrow S^2$. In Section \ref{cohloopstS2} we describe coherent loop states on $S^2$, and show that the constant height coherent loop states recover the standard $|j m \rangle$ basis used by physicists. In Section \ref{inprodsec} we study the inner products of coherent loop states on $S^2$, and derive the formula of Littlejohn and Yu.

\section{The Bergman kernel and coherent loop states} \label{Bergkersec}
In this section we review coherent states, the Bergman kernel, and the construction of coherent loop states in a general setting. Our aim is to express these constructions in an invariant geometric way, to facilitate later computations.

\subsection{Coherent states} \label{cohstatessec}
Let $M$ be a compact complex manifold equipped with a volume form, and let $L$ be a holomorphic line bundle over $M$, equipped with a fibrewise Hermitian inner product $( \cdot, \cdot)_{x \in M}$. We will write the fibrewise norm of a section $s$ of $L$ as $|s(x)|$. 

We also have a {\em global} inner product on the vector space of smooth sections $C^\infty(M, L)$, which we write with angular brackets:
\begin{equation} \label{inn_prod}
 \langle s, s' \rangle = \int_{x \in M} (s(x), s'(x))_x \vol_x.
\end{equation}
Note that our inner products are linear in the second slot and antilinear in the first slot. 

In practice, $L$ will be the dual of a given line bundle $\tau$. We write $P \subset \tau$ for the unit circle bundle of $\tau$. Let $s \in C^\infty(M,L)$, $p \in P$ and $x = \pi(p) \in M$. We use round and square brackets to prevent proliferation of brackets and to conveniently distinguish between the evaluation of $s$ (as a section) at $x \in M$, and the evaluation of $s(x)$ (as a linear functional) at $p \in P$:
\begin{equation} \label{shorthand}
 s(x) \in L_x, \quad s[p] := s(x)(p) \in \mathbb{C}.
\end{equation}
Since $M$ is compact, the vector space $V = H^0(X, L)$ of holomorphic sections of $L$ is finite-dimensional, and is equipped with the inner product \eqref{inn_prod}. For every $p \in P$, we have the {\em evaluation map}
\begin{align*}
	\ev_p : V & \rightarrow \mathbb{C} \\
	s & \mapsto s[p] \, .
\end{align*}

\begin{definition} The {\em coherent state $\psi_p$ localized at $x= \pi(p)$ with base vector $p \in P_x$} is the section in $V$ which represents $\ev_p$, i.e. 
\begin{equation}
 \langle \psi_p, s \rangle = s[p] \text{ for all } s \in V. \label{defncohstateeqn}
\end{equation}
\end{definition}
We will occasionally write $\psi_{x,p}$ when we wish to stress the basepoint $x \in M$. Also, we write
\[
 \hat{\psi}_p = \frac{\psi_p}{\sqrt{\langle \psi_p, \psi_p \rangle}}
\]
for the normalized coherent state. Note that 
\begin{equation} \label{equivariance_coherent}
  \psi_{e^{i \theta} p} = e^{-i \theta} \psi_p \, .
\end{equation}

The coherent state $\psi_p$ can be viewed as the projection onto $V$ of a `$\delta$-section' at $x$, i.e. a `smooth section $\delta_p \in C^\infty(M, L)$' such that 
\[
   \langle s, \delta_p \rangle = s[p] \, .
\]
Of course, $\delta_p$ doesn't actually exist as a smooth section, but we can think of it as a {\em limit} of smooth sections $\delta_p^{(n)}$ satisfying
\[
    \lim_{n \rightarrow \infty} \langle s, \delta^{(n)}_p \rangle = s[p] \quad \text{for all } s \in C^\infty(M, L) \, .
\]
Write
\[
 \Pi : C^\infty(M, L) \rightarrow H^0(M, L)
\]
for the orthogonal projection. Then it is easy to see that
\[
 \lim_{n \rightarrow \infty} \Pi(\delta^{(n)}_p) = \psi_p \, .
\]
So, $\psi_p$ is the best holomorphic approximation to the delta section $\psi_p$. 

The following properties of coherent states are well-known and follow straightforwardly from their definition (see eg. \cite{kirwin2007coherent}).
\begin{lemma} \label{cohstatesprop} The coherent states $\psi_p$ satisfy:
\begin{itemize}
 \item[(a)] (Formula in terms of a basis) Let $e_1, \ldots, e_n$ be an orthonormal basis for $V$. Then:
 \[
  \psi_p = \sum_{i=1}^n \overline{e_i[p]} e_i \, .
 \]
 \item[(b)] (Reproducing property) For every $s \in V$, $x \in M$ and $p \in P_x$, 
 \[
  s[p] = \int_{y \in M} \left(\psi_p(y), s(y)\right)_y \vol_y \, .
   \]
  \item[(c)] (Fibrewise norm at basepoint equals square of global norm) We have
  \[
   | \psi_p (x) | = \langle \psi_p, \psi_p \rangle \quad \text{where } p \in P_x \, .
  \]
  \item[(d)] (Maximally peaked at $x$ amongst unit sections) For each $s \in V$ with $\langle s, s \rangle = 1$, we have:
  \[
    |s(x) | \leq | \hat{\psi}_p(x) | \quad \text{where } p \in P_x.
  \]
\end{itemize}
\end{lemma}

\subsection{Tensor powers}
We can also consider tensor powers $L^k$ of the line bundle $L$. We write $V_k = H^0(X, L^k)$. Given $p \in P$, the coherent state
\[ 
 \psi_p^{(k)} \in V_k
\]
is defined as the unique section in $V_k$ such that
\[
 \langle \psi_p^{(k)}, s \rangle = s[p^{\otimes k}]
\]
for all $s \in V_k$. As $k \rightarrow \infty$, the fibrewise norm of $\psi_{x,p}^{(k)}(y)$ drops off exponentially for $x \neq y$ \cite{seto}:
\[
 |B^{(k)}(x,y)| \leq C e^{- \epsilon \sqrt{k} d(x,y) }.
\]
We will see this explicitly for $M=S^2$ in Corollary \ref{offDiagDropoff}.

\subsection{The Bergman kernel}
\begin{definition} \label{bergman_def} The {\em Bergman kernel} of $L$ is the holomorphic section 
\[
 B \in H^0(\overline{M} \times M, \overline{L} \boxtimes L)
\]
defined as 
\begin{equation} \label{bergman_kernel_defn}
 B(x,y) = \langle \psi_p, \psi_q \rangle \,\overline{p} \otimes q \quad p \in P_x, q \in P_y.
\end{equation}
\end{definition}
Note our convenient notation \eqref{bergman_kernel_defn} for expressing a section of $\overline{L} \boxtimes L$; the formula only depends on $x$ and $y$ (and not the choice of $p$ and $q$) due to the equivariance property \eqref{equivariance_coherent}. 

This definition expresses the Bergman kernel as a global inner product of {\em two} coherent states, but using Lemma \ref{cohstatesprop}(a), we can break symmetry and express it in terms of a {\em single} coherent state
\[
  B(x,y) = \psi_q[p] \, \overline{p} \otimes q \quad p \in P_x, q \in P_y 
\]
or in terms of an orthonormal basis:
\begin{align*}
 B(x,y) &= \sum_{i=1}^n \overline{e_i(x)} \otimes e_i(y) \\
        &= \sum_{i=1}^n \overline{e_i[q]} e_i[p] \, \overline{p} \otimes q \, .
\end{align*}

\subsection{Group action on coherent states} \label{group_act_coh_states}
Suppose that a group $G$ acts on $M$, preserving the volume form, and that this action lifts to a unitary $G$-action on the line bundle $\tau$. Then the vector space $V_k$ is a unitary representation of $G$ via the action
\[
 (g \cdot s)[p] = s[g^{-1} \cdot p]\,.
\]
The action of $G$ on coherent states takes a particularly simple form.
\begin{lemma} \label{groupactionstates} The group $G$ acts on coherent states via $g \cdot \psi^{(k)}_p = \psi_{g \cdot p}$.
\end{lemma}
\begin{proof}
From the defining equation \eqref{defncohstateeqn} of coherent states, we must show that
\[
 \langle \psi_{g \cdot p}^{(k)}, s \rangle = \langle g \cdot \psi_p, s \rangle
\]
for all $s \in V_k$. The left hand side is $s[(g \cdot p)^{\otimes k}]$ by definition, while the right hand side is:
\begin{align*}
 \langle g \cdot \psi_p, s \rangle &= \langle \psi^{(k)}_p, g^{-1} \cdot s \rangle & \text{(By unitarity)} \\
 &= (g^{-1} \cdot s)[p^{\otimes k}] & \text{(By definition)}\\
 &= s[(g \cdot p)^{\otimes k}] & \text{(By definition)}.
\end{align*}
\end{proof} 

\subsection{Coherent loop states}
Recall that the fibrewise inner product on the line bundle $L$ induces a unitary connection $\nabla$ on $\tau$ (the {\em Chern} connection) whose connection 1-form $\beta$ relative to a local holomorphic section $s$ is given by
\[
 \beta = \partial \log H
\]
where $H(x) = (s(x), s(x))_x$. The curvature of $\nabla$ is given locally by $d \beta = \overline{\partial} \partial \log H \in \Omega^{1, 1}(M)$, and we say that the fibrewise inner product is {\em positive} if the symmetric bilinear form $g(X,Y) = \omega(X, JY)$ induced by the real 2-form $\omega = i \curv \nabla$ is positive definite, i.e. a Riemannian metric on $M$, which is a K\"{a}hler metric by construction. From now on we assume this to be the case. 

\begin{definition} Let $\gamma : [0, T] \rightarrow M$ be a loop in $M$. We say that $\gamma$ is a {\em Bohr-Sommerfeld} (or {\em rational}) {\em loop} if its holonomy $\Hol(\gamma)$ is a $k$th root of unity for some $k$.
\end{definition}
This condition means that $\psi^{(k)}(\tildegamma(t))$ forms a loop of vectors in $V_k$, where $\tildegamma$ is a parallel-transported lift of $\gamma$ in $P$. In other words, $\psi^{(k)}_{\tildegamma(T)} = \psi^{(k)}_{\tildegamma(0)}$.
\begin{definition} Let $\gamma(t)$ be a Bohr-Sommerfeld loop in $M$ of order $k$, parameterized by arclength, with parallel-transported lift $\tildegamma$. The {\em coherent loop state} $\Psi^{(k)}_\tildegamma$ of order $k$ associated with $\tilde{\gamma}$ is 
\begin{equation}
 \Psi^{(k)}_\tildegamma = \int_{t=0}^T \psi_{\tildegamma(t)}^{(k)} \, dt \, \in V_k\, . \label{int_defn_coh}
\end{equation}
\end{definition}
We can think of $\Psi^{(k)}_\tildegamma$ as the projection $\Pi(\delta_\tildegamma)$ onto $H^0(M, L^k)$ of the `delta section supported on the loop $\tildegamma$', as in Section \ref{cohstatessec}, in the sense that it is the unique vector in $V_k$ satisfying
\[
 \langle \Psi^{(k)}_\tildegamma, s \rangle = \int_{t=0}^T \psi^{(k)}_{\tildegamma(t)} \, dt
\]
for all $s \in V_k$. This is the viewpoint in \cite{borthwick1995legendrian}.
\begin{remark} Every closed loop $\gamma$ in $M$ is arbitrarily close to a Bohr-Sommerfeld loop.
\end{remark}
\subsection{Group action on coherent loop states} \label{groupactionloopstates}
Note that if $G$ is acting on $M$ and compatibly on $L$ as in Section \ref{group_act_coh_states}, then it will act in a very simple way on coherent loop states:
\begin{align}
 g \cdot \Psi^{(k)}_\tildegamma &= g \cdot \int_0^T \psi^{(k)}_{\tildegamma(t)} \, dt \\
  &= \int_0^T g \cdot \psi^{(k)}_{\tildegamma(t)} \\
  &= \int_0^T \psi^{(k)}_{g \cdot \tildegamma(t)}
\end{align}
where we have used Lemma \ref{groupactionstates} in the last step. Therefore, the formalism of coherent loop states is very natural in the equivariant setting, which is one of the main motivations for this paper.

\subsection{Complex stationary phase formula}
Laplace's method for approximating exponential integrals over an $n$-dimensional compact smooth manifold as $k\rightarrow \infty$ is well-known:
\[
\int_M f(x) e^{-k \phi(x)} \, \vol \sim \left(\frac{2 \pi}{k}\right)^{\frac{n}{2}} \sum_x \frac{f(x) e^{-k \phi(x)}}{|\det H_x(\phi)|^{1/2}}.
\]
Here $\phi(x)$ is assumed to be real, the sum runs over minima $x$ of $\phi$, assumed to be finite, and $H_x(\phi)$ is the Hessian matrix of second partial derivatives at $x$, assumed to be nondegenerate. Similarly, the stationary phase formula for oscillating integrals is (see eg. \cite{guillemin1990geometric}):
\begin{equation} \label{original_sp}
 \int_M f(x) e^{i kS(x)} \vol  \sim \left(\frac{2 \pi}{k}\right)^{\frac{n}{2}} \sum_{x} \frac{f(x)e^{ik S(x)}}{|\det H_x(S)|^{1/2}}e^{\frac{i \pi}{4} \sigma_x},
\end{equation}
The sum runs over the critical points of the phase $S$, assumed to be finite and nondegenerate, and $\sigma_x$ is the signature of the Hessian matrix $H_x(S)$ (the number of positive eigenvalues minus the number of negative eigenvalues). 

When the phase is {\em complex}, the formula becomes a combination of these two methods, and can be found in the article of Pemantle and Wilson \cite{pemantle2010asymptotic}. We will only need the simplest version of their result.
\begin{theorem}[\cite{pemantle2010asymptotic}] \label{pmtheorem} Let $M$ be a compact $n$-dimensional smooth manifold embedded in $\mathbb{C}^n$, let $f$ and $\phi$ be analytic functions on a neighborhood of $M$, and suppose $\Rre(\phi) \geq 0$ on $M$. Let $G$ be the subset of points of $\phi$ on $M$ where $d\phi$ vanishes and where $\Rre(\phi)$ is minimized. Assume that $G$ is finite and that at each $x \in G$, the Hessian is nondegenerate. Then as $k\rightarrow \infty$,
\begin{equation} \label{pemantlewilson}
   \int_M f(x) e^{-k \phi(x)} \vol \sim  \left(\frac{2 \pi}{k}\right)^{\frac{n}{2}} \sum_{x \in G} \frac{f(x)e^{-k \phi(x)}}{\sqrt{\det H_x(\phi)}}
\end{equation}
where the square root of the determinant is defined as the product of the principal square roots of the eigenvalues of $H_x(\phi)$.
\end{theorem}
\begin{remark} \label{roberts_error1} In \cite{roberts1999classical}[Section 3.12 formula (6)], Roberts wrote down an incorrect version of this formula, namely
\begin{equation} \label{roberts_error}
\int_M e^{k \psi} \vol \sim \left(\frac{2 \pi}{k}\right)^{\frac{n}{2}} \sum_x \frac{e^{k \psi(x)}}{\sqrt{-\det H_x(\psi)}}
\end{equation}
for even-dimensional manifolds. See Remark \ref{roberts_error2}.
\end{remark}
Formula \eqref{pemantlewilson} can be a bit confusing since the square root of $\det H_x(\phi)$, defined as the product of the principal square roots of the eigenvalues, is not purely a function of $\det H_x(\phi)$. For our purposes it is more convenient to write this formula in the form \eqref{original_sp}, where the square root in the denominator is real, and all the phase information has been made explicit in the numerator.
\begin{lemma}[Complex stationary phase lemma] \label{csplem}  Writing $S = -i\phi$, and under the same assumptions as in Theorem \ref{pemantlewilson}, we have:
\begin{equation} \label{new_sp}
 \int_M f(x) e^{i kS(x)} \, \vol  \sim \left(\frac{2 \pi}{k}\right)^n \sum_{x} \frac{f(x)e^{ik S(x)}}{|\det H_x(S)|^{1/2}}e^{i\sum_j\left(\pi/4 - \alpha_j/2\right)}.
 \end{equation}
Here, $x$ runs over the stationary points (that is, where $dS$ vanishes and $\Rre(\phi)$ is a maximum), and $\alpha_j$ is the principal argument of the $j$th eigenvalue of $H_x(S)$ .
\end{lemma}
\begin{proof}
Let the principal arguments of the eigenvalues $\lambda_i$ of $H_x(\phi)$ be $\theta_i$. We must have $\Rre(\lambda_i) > 0$, so $-\pi/2 < \theta_i < \pi/2$. Therefore, the eigenvalues of $H_x(S) = iH_x(\phi)$ have principal arguments
\[
 \alpha_i = \theta_i + \pi/2
\]
and so
\[
\frac{1}{\sqrt{\det H_x (\phi) }} = \frac{e^{\sum_i (\pi/4 - \alpha_i/2)}}{|\det H_x(S)|} \,.
\]
\end{proof}
Observe that when $S$ is real, \eqref{new_sp} reduces to \eqref{original_sp}.

\section{Coherent states on $S^2$} \label{cohstatesS2}
In this section we discuss coherent states on $S^2$ from a geometric perspective.

\subsection{The tautological line bundle over $S^2$}
Let $M=\mathbb{CP}^1$, the space of 1d subspaces of $\mathbb{C}^2$, equipped with the volume form coming from the Fubini-Study metric. So,
\[
 \int_{\mathbb{CP}^1} \vol = 2 \pi.
\] 
Let $\tau$ be the tautological line bundle over $\mathbb{CP}^1$, whose fibre at a line $l \in \mathbb{CP}^1$ is the line $l \subset \mathbb{C}^2$ itself. Then $\tau$ inherits a fibrewise inner product from $\mathbb{C}^2$. The unit circle bundle $P \subset \tau$ is the 3-sphere $S^3 \subset \mathbb{C}^2$. 

Since we are interested in three-dimensional geometric formulas for angular momentum, we want to think of $M$ as being $S^2$ instead of $\mathbb{CP}^1$. So, we identify $\mathbb{CP}^1$ with $S^2$ by first identifying $S^2$ with the extended complex plane $\mathbb{C} \cup \{\infty\}$ via stereographic projection from the {\em south} pole (so that the map is orientation-preserving) as in Figure \ref{sphericalproj}, and then identifying the extended complex plane with $\mathbb{CP}^1$ via $z \mapsto [z : 1]$. In terms of latitude and longitude coordinates on $S^2$, this means:
\begin{align}
   S^2 & \xrightarrow{\cong} \mathbb{CP}^1 \\
   (\theta, \phi) & \mapsto [\sin \frac{\theta}{2} e^{i \phi} : \cos \frac{\theta}{2} ] \, .
\end{align}
Then, then the projection map $S^3 \rightarrow S^2$ is the Hopf fibration. Note that the formula 
\begin{equation} \label{usec}
  u(\theta, \phi) = \begin{pmatrix} \sin \frac{\theta}{2} e^{i \phi} \\
   \cos \frac{\theta}{2}
   \end{pmatrix}
\end{equation}
gives a smooth trivialization of the Hopf fibration over $S^2 \setminus \{(0,0,-1)\}$. 

\begin{figure}
\centering
\begin{tikzpicture}[scale=2,>= stealth]
\draw (0,0) circle (1);
\draw[->] (0,0,0) coordinate (O) -- (1.5,0,0) coordinate(X) node[pos=1.1] {Im};
\draw[->]  (O) -- (0,1.5,0) coordinate(Y);
\draw[->]  (O) -- (0,0,2.5) coordinate(Z) node[pos=1.05] {Re};
\filldraw[black] (0,-1) circle (0.027);
\draw[dashed](0,-1)-- (0.5,0.5);
\draw[-](0.5,0.5)-- (1.,1.2);
\draw (1.1,1.3) circle (.15);
\filldraw[black] (1.24,1.35) circle (0.027) node[anchor=west]{$u(\theta,\phi)$};
\filldraw[black] (0.5,0.5) circle (0.027) node[anchor=west] {$x(\theta,\phi)$};
\draw[dashed](0,0)-- (0.5,0.5);
\draw[dashed](0,0)-- (0.26,-0.2);
\filldraw[black] (0.26,-.2) circle (0.027) node[anchor=west] {$z$};
\coordinate (a) at (0,0);
\coordinate (b) at (.2,.2);
\coordinate (c) at (0,.3);
\coordinate (d) at (-.5,-.5);
\coordinate (e) at (.65,-.5);
\draw pic[draw,angle radius=.5cm,"$\theta$" ] {angle=b--a--c};
\draw pic[draw,angle radius=.5cm,"$\phi$" ] {angle=d--a--e};
\end{tikzpicture}

\caption{\label{sphericalproj}Spherical projection from the south pole identifies a point $x(\theta, \phi) \in S^2$ with $z=\tan(\theta/2) e^{i \phi} \in \mathbb{C}$. Also sketched is the fiber $P_x \subset S^3$ of the Hopf fibration at $x$, and the point $u(\theta, \phi) \in P_x$.  }
\end{figure}
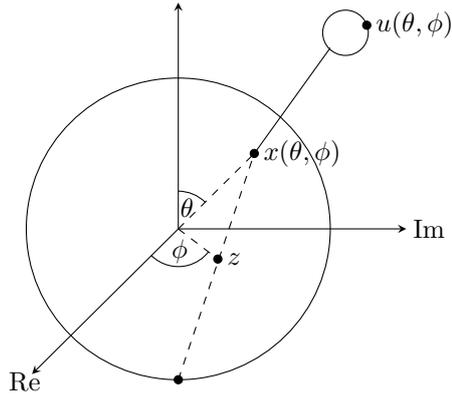
\subsection{The dual of the tautological line bundle}
The line bundle $L:= \tau^\vee$ is defined as the dual of $\tau$. The space of holomorphic sections of $L^k$ canonically identifies with the $(k+1)$-dimensional vector space of homogenous polynomials of degree $k$ in two variables $Q_1,Q_2$ via the evaluation map, where $Q_1(q_1, q_2) = q_1$ and $Q_2(q_1, q_2) = q_2$:
\begin{align}
	\Poly_k (Q_1,Q_2) & \rightarrow V_k \\
	f & \mapsto ( (q_1,q_2)^{\otimes k}  \mapsto f(q_1,q_2) )  .
\end{align}
With respect to this identification, an orthonormal basis for $V_k$ is given by
\begin{equation} \label{orth_basis}
 e_a = \sqrt{\frac{k+1}{2 \pi} \binom{k}{a} }\, Q_1^a Q_2^{k-a}, \quad a=0, 1, \ldots, k.
\end{equation}

\subsection{Coherent states on $S^2$ in terms of $S^3$}
The following formulas express the coherent states and Bergman kernel for $L^k \rightarrow S^2$ in terms of the Hopf fibration $S^3 \rightarrow S^2$. The proof is a simple calculation.
\begin{lemma} \label{coh_state_S2} The coherent state on $S^2$ with base vector $p \in S^3$ is given by:
\[
 \psi^{(k)}_p [q] = \frac{k+1}{2 \pi} \langle p, q \rangle_{\mathbb{C}^2}^k
\]
That is, 
\begin{equation} \label{inprodcohS2}
 \langle \psi^{(k)}_p, \psi^{(k)}_q \rangle = \frac{k+1}{2 \pi} \langle q, p \rangle^k \, ,
\end{equation}
Moreover, the Bergman kernel is given by:
\[
 B^{(k)} (x,y) = \frac{k+1}{2\pi} \langle q, p \rangle^k \, \overline{p} \otimes q
\]
\end{lemma}
As a result, we have:
\begin{corollary}\label{offDiagDropoff} The magnitude of the off-diagonal Bergman kernel on $S^2$ is given by
\begin{equation} \label{bergmanks2}
 |B^{(k)}(x,y)| = \frac{k+1}{2 \pi} \cos^k(\frac{\theta}{2})
\end{equation}
where $\theta$ is the geodesic distance between $x$ and $y$ on $S^2$.
\end{corollary}
\begin{proof}
By rotational symmetry, it suffices to compute this when $x$ is the north pole $(0,0,1) \in S^2$ and $y$ is at $\phi = 0, \theta = \theta$. We can use formula \eqref{usec} to give us points $p \in P_x$ and $q \in P_y$ as $p = (0, 1)$ and $q = (\sin\frac{\theta}{2}, \cos\frac{\theta}{2})$. So, $\langle p, q \rangle = \cos \frac{\theta}{2}$ and the formula follows.
\end{proof}
Note that as promised, for $x \neq y$, $|B^{(k)} (x,y)| \rightarrow 0$ exponentially as $k \rightarrow \infty$. 

\begin{remark} The coherent state attached to $(0,1)^T \in S^3$ is proportional to the spin-up state $|j=1/2, m=1/2\rangle$ while the coherent state attached to $(1,0)^T \in S^3$ is proportional to the spin-down state $|j=1/2, m=-1/2\rangle$. See Section \ref{loopsareangs}.
\end{remark}

\subsection{Coherent states on $S^2$ in terms of complex variables}
Instead of thinking of the coherent state geometrically as a function on $S^3$ as in Lemma \eqref{coh_state_S2}, we can also think of it (as is more common, see eg. \cite{kirwin2007coherent}) as a degree $k$ homogenous polynomial $\psi_z^{(k)}$ in $w_1$ and $w_2$, parameterized by $z \in \mathbb{C}$. To do this, we will need to use the local trivialization $u(z)$ of the Hopf fibration from \eqref{usec}: 
\begin{align}
  \psi_z^{(k)}[w] &= \frac{k+1}{2 \pi} \langle u(z), q \rangle^k \\
   &= \frac{k+1}{2 \pi} \frac{1}{\sqrt{1 + |z|^2}} (\bar{z}q_1 + q_2)^k \\
   &= \frac{k+1}{2 \pi} \frac{1}{\sqrt{1 + |z|^2}}\sum_{r=0}^k \left( \begin{array}{c} k \\ r \end{array} \right) \bar{z}^r q_1^r q_2^{k-r}
\end{align}
Therefore, 
\[
 \psi_z^{(k)} = \sqrt{\frac{k+1}{2 \pi}} \frac{1}{\sqrt{1 + |z|^2}}\sum_{r=0}^k \bar{z}^r e_r \, .  
\]

\subsection{Connections and curvature}
Let us compute the induced Chern connection on the Hermitian tautological line bundle $\tau$, whose unit circle bundle $P$ is the Hopf fibration $S^3 \rightarrow S^2$. We take as local unitary section $u(\theta, \phi)$ given in \eqref{usec}. Then it is straightforward to calculate that
\[
  \nabla_X u = i \sin^2 \frac{\theta}{2} \, d\phi(X) u\,.
\]
Therefore, if we define the local connection 1-form $\alpha$ on $\tau$ by
\[
 \nabla_X u = i \alpha(X) u, 
\]
then
\begin{equation} \label{con_form_tau}
 \alpha = \sin^2 \frac{\theta}{2} \, d \phi  \, .
\end{equation}
Therefore, the local connection 1-form on $L = \tau^\vee$ is $-\alpha$, and $\curv (L, \nabla) = d(-i \alpha)$. So, the K\"{a}hler 2-form $\omega$ on $S^2$ is given by
\begin{align*}
 \omega &= i \curv (L, \nabla) \\
 &= i (-i d \alpha) \\
 &= \frac{1}{2} \sin \theta d\theta \wedge d\phi 
\end{align*}
which makes sense, as it is the Fubini-Study area form on $\mathbb{CP}^1$ transplanted to $S^2$. Note that $\omega$ is one half the usual area form on $S^2$. 

From these calculations, we conclude the following facts about parallel transport in the Hopf fibration $S^3 \rightarrow S^2$. The second statement follows naturally from Stokes' theorem.
\begin{lemma} \label{hollemma} If $\gamma(t)$ is a path in $S^2$, then its parallel-transported lift $\tildegamma(t)$ in $S^3$, starting at $\tildegamma(0) = u(\gamma(0))$, is given by
\begin{equation} \label{partransportformula}
 \tildegamma(t) = e^{-i \int\limits_{\gamma} \alpha}  \, u(\gamma(t)) \, .
 \end{equation}
Moreover, if $\gamma$ is an oriented closed loop in $S^2$ without self-intersections, then the holonomy around $\gamma$ is given by
\begin{equation} \label{holonomy}
 \Hol(\gamma) = e^{-i A/2}
\end{equation}
where $A = \partial \gamma$ is the spherical area enclosed by $\gamma$, determined by the standard right-hand-rule (turning counterclockwise from $\gamma'(t)$ leads into $A$).
\end{lemma}
\begin{remark} Note that we have expressed the holonomy in terms of the {\em standard} area on $S^2$, which is twice the {\em symplectic} area.
\end{remark}

\section{Coherent loop states on $S^2$} \label{cohloopstS2}
From Lemma \ref{holonomy}, we see that a loop $\gamma$ in $S^2$ without self-intersections is a Bohr-Sommerfeld loop of order $k$ precisely when the quantization condition
\[
 e^{ikA/2} = 1
\]
is satisfied, where $A$ is the area enclosed by the loop. Note that since the area of the sphere is $4 \pi$, this condition is not sensitive to which of the two areas `enclosed by $\gamma$' is being used.

\subsection{Constant height loop states}

In particular, this allows us to compute the Bohr-Sommerfeld loops having fixed height on $S^2$.
\begin{lemma} A loop at fixed height $\z = \cos\theta$ on $S^2$ is a Bohr-Sommerfeld loop of order $k$ if and only 
\begin{equation} \label{discrete_thetas}
 \cos \theta = 1- \frac{2a}{k}, \quad a = 0, 1, \ldots, k.
\end{equation}
\end{lemma}
\begin{remark} We use $z$ for points in $\mathbb{C}$, and $\z$ for the vertical coordinate in $\mathbb{R}^3$.
\end{remark}
\begin{definition} \label{standardloop} Let 
\[
 \gamma_m(\phi) = (\sin \theta_m \cos \phi, \sin \theta_m \sin \phi, \cos \theta_m)
\]
 be the loop on $S^2$ at the discrete height $\cos \theta_m = 2m/k$, $m \in \{-k/2, -k/2+1, \ldots, k/2\}$. We define its {\em standard parallel-transported lift} to $S^3$ as
\begin{equation} \label{gammaphi}
 \tildegamma_m(\phi) = e^{i\sin^2\left(\frac{\theta_m}{2}\right) \phi} u(\theta_m, \phi) \, .
\end{equation}
To fix possible ambiguities at the north and south poles of $S^2$ ($\z = 1$ and $\z=-1$ respectively) where the loops are constant, we set $\tildegamma_{\z=1} = (0, 1)^T \in S^3$ and $\tildegamma_{\z=-1} = (1, 0)^T \in S^3$ respectively. 
\end{definition}
It is sometimes useful to be able to index these loops as $\tildegamma^a$ using the mathematician's numbering $a= 0, 1, \ldots k$ instead of the physicist's numbering $m=-j, -j+1, \ldots, j$ where $j=k/2$. The correspondence is
\[
  \tildegamma^a \equiv \tildegamma_m \quad \text{where } a=\frac{k}{2} - m \, .
\]

\subsection{Loop states are angular momentum eigenstates} \label{loopsareangs}
In this section, we will identify the coherent loop states\footnote{The prefactor of $\frac{\sin \theta_m}{\sqrt{2}}$ is there to account for the arclength.}
\begin{equation} \label{preciseformcohstate}
 \Psi^{(k)}_{\tildegamma_m} = \frac{\sin \theta_m}{\sqrt{2}} \int_0^{2 \pi} \psi^{(k)}_{\tildegamma_m(\phi)} d \phi
\end{equation}
with the spin angular momentum eigenstates $|j m \rangle$ familiar to physicists, where $j = \frac{k}{2}$:
\begin{equation} \label{prove_eigenstate}
 \Psi^{(k)}_{\tildegamma_{m}} \, \propto \, \big|jm\rangle
\end{equation}
See Figure \ref{local_state}. 

Recall that the state $|jm \rangle$ is defined as an eigenstate of the angular momentum operator $\hat{J}_z$, where
\[
  iJ_z = \frac{1}{2} \left[ \begin{array}{cc} 1 & 0 \\ 0 & -1 \end{array} \right] 
\]
is the element in the Lie algebra $\mathfrak{su}(2)$ generating rotations about the $\z$-axis. That is,
\begin{equation} \label{angstateliealg}
 \hat{J}_z |jm \rangle = m  |jm \rangle
\end{equation}
in units where $\hbar = 1$. In terms of the group element 
\begin{equation} \label{Uzphi}
 U_z(\Delta \phi) = e^{i \Delta\phi J_z} = \left[ 
 \begin{array}{cc} e^{i \frac{\Delta \phi}{2}} & 0 \\
 0 & e^{-i \frac{\Delta \phi}{2}} \end{array} \right] \, \in \SU(2)
\end{equation}
covering the rotation matrix $R_z (\Delta \phi)$ in $\SO(3)$, \eqref{angstateliealg} is equivalent to
\[
 U_z (\Delta \phi) | jm \rangle = e^{im \Delta \phi} | jm \rangle \,.
\]
To prove \eqref{prove_eigenstate}, we need to compute the difference between parallel-transporting in the Hopf fibration $S^3 \rightarrow S^2$ in the $\phi$ direction by $\Delta \phi$ and acting with $U_z (\Delta \phi)$. These are not the same, except at $\z = 0$; the following result records the difference.
\begin{lemma} $U_z(\Delta \phi)  \cdot \tildegamma_m(\phi) = e^{-i \z_m \frac{\Delta\phi}{2}} \tildegamma(\phi + \Delta\phi)$.
\end{lemma}
\begin{proof} This follows from simply substituting in \eqref{Uzphi},\eqref{gammaphi} and  \eqref{usec} respectively.
\end{proof}
This immediately gives us a geometric proof of the following. 
\begin{theorem}\label{first_proof} The coherent loop state $\Psi^{(k)}_{\tildegamma_m}$ is an eigenstate of the angular momentum operator $\hat{J}_\z$ with eigenvalue $m$:
\[
 U_{\z} (\Delta \phi) \cdot \Psi^{(k)}_{\tildegamma_m} = e^{i m \Delta\phi}\, \Psi^{(k)}_{\tildegamma_m}.
\]
\end{theorem}
\begin{proof} 
Write $U =  U_{\z} (\Delta \phi)$. Then, ignoring the arclength prefactor in \eqref{preciseformcohstate} which plays no role, 
\begin{align*}
 U  \cdot \Psi^{(k)}_{\tildegamma_m}  &=  U  \cdot \int_{\phi=0}^{2 \pi} \psi^{(k)}_{\tildegamma_m(\phi)} \, d \phi \\
&= \int_{\phi=0}^{2 \pi}  U \psi^{(k)}_{\tildegamma_m(\phi)} \, d \phi  \\
&= \int_{\phi=0}^{2 \pi} \psi^{(k)}_{U \tildegamma_m(\phi)} \, d \phi  \\
&= \int_{\phi=0}^{2 \pi} \psi^{(k)}_{e^{-iz_m\frac{\Delta\phi}{2}}  \tildegamma_m(\phi + \Delta\phi)} \, d \phi  \\
&= \int_{\phi=0}^{2 \pi}e^{im \Delta\phi}  \psi^{(k)}_{\tildegamma_m(\phi+\Delta\phi)} \, d \phi  \\
&= e^{im\Delta\phi} \int_{\phi=0}^{2 \pi} \psi^{(k)}_{\tildegamma_m(\phi+\Delta\phi)} \, d \phi  \\
&= e^{im\Delta\phi} \Psi^{(k)}_{\tildegamma_m} \, .
\end{align*}
\end{proof}

We now give a different proof of Theorem \ref{orth_basis}, by performing an explicit calculation in the framework of homogenous polynomials. First we show that the coherent loop states $\Psi_{\tildegamma^a}^{(k)}$ are proportional to the orthonormal basis $e_a$ for $V_k$ from \eqref{orth_basis}.  
\begin{lemma} \label{Psi_explicit} $\Psi_{\tildegamma^a}^{(k)} = c_a e_a$, where $
c_a = \sqrt{\pi(k+1)\binom{k}{a}} \sin^a\left(\frac{\theta_a}{2}\right) \cos^{k-a}\left(\frac{\theta_a}{2}\right) \sin\theta_a.$
\end{lemma}
\begin{proof} We compute the integral \eqref{int_defn_coh}:
\begin{align*}
    \Psi^{(k)}_{\tildegamma^a}[q] &= \frac{\sin\theta_a}{\sqrt{2}} \int_0^{2\pi} \psi^{(k)}_{\tildegamma^a(\phi)} \, d\phi \\
    &= \frac{\sin\theta_m}{\sqrt{2}} \int_0^{2\pi} \psi^{(k)}_{e^{-i \sin^2\left(\frac{\theta_a}{2}\right)\phi} u(\theta_a, \phi)} \, d\phi \\
 &= \frac{k+1}{2 \pi} \frac{\sin\theta_a}{\sqrt{2}} \int_0^{2\pi} e^{k i \sin^2\left(\frac{\theta_a}{2}\right)\phi} \, \left\langle \left(\sin\left(\frac{\theta_m}{2}\right) e^{i \phi}, \cos\left(\frac{\theta_m}{2}\right), \, (q_1, q_2) \right) \right\rangle^k \, d\phi \\    
    &= \frac{k+1}{2 \pi} \frac{\sin\theta_a}{\sqrt{2}} \sum_{r=0}^k \sin^r\left(\frac{\theta_a}{2}\right) \cos^{k-r}\left(\frac{\theta_a}{2}\right) q_1^r q_2^{k-r} \int_0^{2\pi} e^{i(a-r)\phi} \, d\phi \\
    &= \frac{k+1}{\sqrt{2}} \sin\theta_a \sin^a\left(\frac{\theta_a}{2}\right) \cos^{k-a}\left(\frac{\theta_a}{2}\right) \\
    &= c_a e_a [q].
\end{align*}

\end{proof}
\begin{lemma}\label{eaeigenstate} $\hat{J}_z e_a = \left(\frac{k}{2} - a\right) e_a.$
\end{lemma}
\begin{proof} To compute the action of $\xi = i\hat{J}_z$ on a state in $V_k$, 
\begin{equation} \label{Jzstate}
 (\xi \cdot s)(p) = \frac{d}{dt}\Bigr|_{t=0} U_z(t) s \left( R_z(-t) p\right)\,
\end{equation}
one might use the quantization formula of Kostant and Souriau \cite{dynamiques170b, souriau1997structure},
\[
 \xi \cdot s = -\nabla_{\phi} s + \mu(J_z) s \, ,
\]
where $\mu$ is the moment map. Instead, we will simply use the definition \eqref{Jzstate} directly:
\begin{align*}
	\frac{d}{dt}\Bigr|_{t=0} U_z(t) Q_1^aQ_2^{k-a} \left( R_z(-t) p\right)[q] &= \frac{d}{dt}\Bigr|_{t=0} Q_1^a Q_2^{k-a} \left( e^{-i\frac{t}{2}} q_1, \, e^{-\frac{t}{2}} q_2 \right) \\
	&= \frac{d}{dt}\Bigr|_{t=0} e^{i(k-2a) \frac{t}{2}} q_1^k q_2^{k-a} \\
	&= i(k-2a) q_1^k q_2^{k-a}.
\end{align*}
\end{proof}
\begin{proof}[Second proof of Theorem \ref{first_proof}]
Lemma \ref{Psi_explicit} says that $\Psi_{\tildegamma^a}$ is a multiple of $e_a$, while Lemma \ref{eaeigenstate} says that $e_a$ can be identified with the eigenstate $|jm \rangle$ where $m = \frac{k}{2} - a$.
\end{proof}

\section{Inner products of loop states on $S^2$} \label{inprodsec}
In this section we give our own self-contained proof (in our current setting) of Borthwick, Paul and Uribe's result \cite{borthwick1995legendrian} (Theorem \ref{BPUTheorem}) on the asymptotics of the inner products of coherent loop states. We use this to write down a spherical area formula for the asymptotics of very general coherent loop states on $S^2$, and apply this to derive Littlejohn and Yu's formula for the asymptotics of the Wigner matrix elements.

\subsection{Asymptotics of norm of coherent loop states}
The following formula recovers \cite[Theorem 4.4a]{borthwick1995legendrian} in our setting.
\begin{lemma} As $k \rightarrow \infty$, and keeping $\cos \theta_m = \frac{2m}{k}$ fixed, we have
\begin{equation} \label{norm_lemma}
 \left\langle \Psi^{(k)}_{\tildegamma_m}, \Psi^{(k)}_{\tildegamma_m} \right\rangle \sim \sqrt{\frac{k}{\pi}}T
\end{equation}
where $T=\frac{2\pi\sin\theta_m}{\sqrt{2}}$ is the arclength of $\gamma_m$.
\end{lemma}
\begin{proof} From Lemma \ref{Psi_explicit}, and using the relation $m = k/2 - a$, we have
\begin{align*}
\langle \Psi^{(k)}_{\tildegamma_m},\Psi^{(k)}_{\tildegamma_m}\rangle &= \pi(k+1)\binom{k}{a} \sin^{2a}\left(\frac{\theta_a}{2}\right) \cos^{2(k-a)}\left(\frac{\theta_a}{2}\right) \sin^2\theta_a. \\
&\sim \pi(k+1)\sin^{2}\left(\theta_a\right) \sqrt{\frac{k}{2\pi a(k-a)}} \\
&\sim \pi k \sin^{2}\left(\theta_a\right)\frac{1}{\sqrt{2 \pi}} \frac{2}{\sin \theta_a} \\
&= \sqrt{\frac{k}{\pi}}T
\end{align*}
where the second line follows from Stirling's approximation $n!\sim\sqrt{2\pi n}(\frac{n}{e})^{n}$.
\end{proof}

\subsection{Warm-up example}
For instructive purposes, let us apply the technology of coherent loop states to the same warm-up example as in \cite[Theorem 8]{roberts1999classical}. We will calculate the asymptotics of the small Wigner $d$-matrix, for integral $j$ (i.e. $k = 2j$ is even), on the `equator states':
\begin{equation} \label{innprodwex}
 d^j_{00}(\beta) = \langle j0 | \, U_y(\beta) \, | j0 \rangle .
\end{equation}
Note that this matrix element can be computed exactly \cite{biedenharn1984angular} as 
\[
d^j_{00}(\beta) = \frac{2}{\sqrt{2 \pi j \sin \beta}} P_j(\cos \beta)
\]
where $P_n (x)$ is the $n$th Legendre polynomial.

As explained in the Introduction, {\em without} the technology of coherent loop states, in the framework of geometric quantization, one would write $|j 0 \rangle$ as a multiple of the holomorphic section
\[
s = Q_1^{j} Q_2^{j}
\]
as in \eqref{orth_basis}, and then one would express the inner product $\langle s, U_y(\beta) s \rangle$ as an integral over $S^2$, 
\[
 \frac{\langle s, U_y(\beta)s\rangle}{\langle s, s\rangle} = \frac{1}{\langle s, s\rangle} \int_{x \in S^2} \left(s(x), (U_y(\beta)(s))(x)\right)_x \vol_x
\]
whose asymptotics can then be computed using the stationary phase principle. This was the approach in \cite[Theorem 8]{roberts1999classical}.

{\em With} the technology of coherent loop states, we rather express $|j0 \rangle$ as a multiple of the equator coherent loop state $\Psi^{(k)}_\tildegamma$, where $k=2j$ is even. This gives two advantages. Firstly, $SU(2)$ acts very naturally on coherent loop states, as we saw in Section \ref{groupactionloopstates},
\[
 U_y(\beta) \Psi^{(k)}_{\tildegamma} = \Psi^{(k)}_{\tildesigma}
\]
where $\tildesigma(t) = U_y(\beta) \tildegamma(t)$. Secondly, it means that when we compute the normalized inner product
\begin{equation} \label{numdenom}
\langle j0 | \, U_y(\beta) \, | j0 \rangle = \frac{\langle \Psi_\tildegamma^{(k)}, \Psi^{(k)}_\tildesigma\rangle}{\langle \Psi^{(k)}_\tildegamma, \Psi^{(k)}_\tildegamma \rangle} 
\end{equation}
we can exchange integrals and express the inner product of coherent {\em loop} states as an integral over the {\em torus} $S^1 \times S^1$ (not over $S^2$) of the global inner product of {\em coherent states}, for which we have an elegant geometric formula \eqref{inprodcohS2}:
\begin{align}
 \langle \Psi_\tildegamma^{(k)}, \Psi^{(k)}_\tildesigma\rangle &= \frac{1}{2} \int_{0}^{2 \pi} \int_{0}^{2 \pi} \bigl\langle \psi^{(k)}_{\tildegamma(s)}, \psi^{(k)}_{\tildesigma(t)} \bigr\rangle \, ds dt \\
 &= \frac{k+1}{4 \pi} \int_{0}^{2\pi} \int_0^{2\pi}  \langle \tildesigma(t), \tildegamma(s) \rangle^k_{\mathbb{C}^2}   \label{secondstop}
\end{align}
(The factor of $\frac{1}{2}$ is there to account for arc length on the equator). This integral is very amenable to calculation.

\begin{proposition} \label{warmupprop} For $\beta \in (0, \pi)$, the small Wigner $d$-matrix for equatorial states at large even $k$ has the following asymptotics:
\[
 d^{\frac{k}{2}}_{00}(\beta) \sim 
\frac{2}{\sqrt{\pi k\sin\beta}}\cos\left((k+1) \beta/2-\pi/4\right) .
\]
\end{proposition}
\begin{proof}
We need the formula for the equatorial Bohr-Sommerfeld loop from \eqref{gammaphi},
\[
 \tildegamma(t) = \frac{1}{\sqrt{2}} \begin{pmatrix} e^{i t/2} \\ e^{-i t/2} \end{pmatrix}
\]
as well as the rotation matrix $U_y(\beta) \in SU(2)$:
\[
 U_y(\beta) = \begin{pmatrix} \cos \beta/2 & \sin\beta/2 \\ -\sin\beta/2 & \cos \beta/2 \end{pmatrix}
\]
Substituting these in, \eqref{secondstop} becomes:
\begin{align}
 \langle \Psi_\tildegamma^{(k)}, \Psi^{(k)}_\tildesigma\rangle &=\frac{k+1}{4\pi}\int_{0}^{2\pi}\int_{0}^{2\pi}\left(\cos\beta/2\cos(s-t)/2 + i\sin\beta/2 \sin(s+t)/2)\right)^{k} ds dt \nonumber \\
 &=\frac{k+1}{4 \pi} \int_0^{2\pi} \int_0^{2 \pi} e^{ikS} \, ds dt \label{firststatph}
\end{align}
where
$$
S = -i \log\left(\cos\beta/2\cos(s-t)/2 + i\sin\beta/2 \sin(s+t)/2\right).
$$
Let us apply the complex stationary phase principle (Lemma \ref{csplem})  to the integral \eqref{firststatph}. Since $|\langle \tildesigma(t), \tildegamma(s) \rangle| \leq 1$, the integrand satisfies the prerequisites of Theorem \ref{csplem}. The stationary points of $S$ occur at $(s,t) = (\frac{\pi}{2}, \frac{\pi}{2})$ and $(s,t) = (\frac{3\pi}{2}, \frac{3\pi}{2})$, which are precisely the two points of intersection of the loop $\gamma$ and $\sigma$. The value of $e^{ikS}$ at each critical point equals $e^{ik \beta/2}$ and $e^{-ik \beta/2}$ respectively. At $x = (\frac{\pi}{2}, \frac{\pi}{2})$, the Hessian computes as:
\[
H_x = \begin{pmatrix}
S_{ss} & S_{st} \\
S_{ts} & S_{tt}
\end{pmatrix}= \frac{i}{4}
\begin{pmatrix}
1 & -e^{-i\beta}\\
-e^{-i\beta} & 1
\end{pmatrix} .
\]
The eigenvalues are
\[
 \lambda_1 = \frac{1}{2} \sin(\beta/2)e^{i(\pi- \beta/2)}, 
 \quad \lambda_2 = \frac{1}{2} \cos(\beta/2)e^{i\left(\pi/2 - \beta/2\right)},
 \]
so that $|\det H_x| = 1/8 \sin \beta$ and the principal angles are $\alpha_1 = \pi - \beta/2$ and $\alpha_2 = \pi/2 -\beta/2$. So, the various ingredients of the stationary phase contribution from the critical point are:
\begin{align}
 \frac{2\pi}{k} f(x) e^{ikS}\frac{e^{i \sum_{j}\left( \pi/4 - \alpha_j/2\right)}}{\sqrt{|\det H_x|}} &= \frac{2\pi}{k} \frac{(k+1)e^{ik\beta/2}}{4\pi} \frac{e^{i(\beta/2 - \pi/4)}}{\sqrt{\sin(\beta)/8}} \\
 &\sim \sqrt{2} e^{ik \beta/2} \frac{e^{i( \beta/2 - \pi/4)}}{\sqrt{\sin \beta}} \label{statc1}.
 \end{align}
A similar calculation shows that the other critical point delivers the conjugate  contribution
\begin{equation}
 \sqrt{2} e^{-ik \beta/2} \frac{e^{-i(\beta/2 - \pi/4)}}{\sqrt{\sin \beta}} \, . \label{statc2}
\end{equation}
Putting it all together, the asymptotics of \eqref{firststatph} is
$$
\sqrt{\frac{2}{\sin\beta}} 2 \cos\{(k+1)\beta/2-\pi/4\}.
$$
Dividing by the norm squared of the equator state from \eqref{norm_lemma},
\[
\langle \Psi^{(k)}_\tildegamma, \Psi^{(k)}_\tildegamma \rangle \sim \sqrt{2 \pi k},
\]
gives the final result.  
\end{proof}

\begin{remark} \label{roberts_error2} In \cite[Theorem 8]{roberts1999classical}, the asymptotic formula was computed incorrectly as 
\[
\frac{2}{\sqrt{\pi k\sin\beta}} \cos\left((k+1)\beta/2 + \pi/4 \right) \,.
\]
The source of the error was the error in the complex stationary phase formula \eqref{roberts_error}. This only affected the warmup example and not the main results of \cite{roberts1999classical}.
\end{remark}

\subsection{The general case}
In fact, the stationary phase calculation we did in the warm-up example is all we need to derive the theorem of Borthwick, Paul and Uribe \cite[Theorem 4.4b]{borthwick1995legendrian}in our setting. 
\begin{theorem} \label{mainthm1} Let $\gamma$ and $\sigma$ be curves in $S^2$ intersecting in finitely many points, with parallel transported lifts $\tildegamma$ and $\tildesigma$ in $S^3$ respectively. Then, when $k \rightarrow \infty$ through Bohr-Sommerfeld values,  
\[
 \langle \Psi^{(k)}_{\tildegamma}, \Psi^{(k)}_{\tildesigma} \rangle \sim \sqrt{2}\sum_{x\in\gamma\cap\sigma}\frac{\omega_x^{k}e^{i \orr_x (\theta_x/2-\pi/4)}}{\sqrt{\sin\theta_x}}
\]
where:
\begin{itemize}
   \item At each intersection point $x$, $\omega_x = \tildegamma_x/\tildesigma_x \in U(1)$, and
   \item $\theta_x \in (0, \pi)$ is the angle between $\gamma'$ and $\sigma'$ in $T_x S^2$ (always positive), and
   \item $\orr_x \in \{+1, -1\}$ is the {\em orientation} of $\theta_x$ (it is $+1$ if rotating $\gamma'$ to $\sigma'$ agrees with the orientation of $S^2$, i.e. counterclockwise when viewed from outside $S^2$, and negative otherwise).
\end{itemize}
\end{theorem}
\begin{proof}
We need to calculate the asymptotics of the following integral:
\begin{align*}
  \langle \Psi^{(k)}_{\tildegamma}, \Psi^{(k)}_{\tildesigma} \rangle &= \frac{k+1}{2\pi} \int_0^S \int_0^T \langle \tildesigma(t), \tildegamma(s) \rangle_{\mathbb{C}^2}^k \, ds dt \\
 &=  \frac{k+1}{2\pi} \int_0^S \int_0^T e^{ikS} ds dt.
\end{align*}
To do this, we simply need to consider the calculation of the warm-up example in Proposition \eqref{warmupprop} in a more geometrically invariant way. Now, in that example, $\gamma$ was the equator loop, $\tildegamma$ was a parallel transported lift of it in $S^3$, $\tildesigma = U_y(\beta) \tildegamma$, and  $\beta \in (0, \pi)$. We found that the critical points occurred at the points of intersection of the shadow curves $\gamma$ and $\sigma$ on $S^2$. 

In that example, $\gamma$ and $\sigma$ were geodesics. So, by symmetry, we know that for {\em any} Bohr-Sommerfeld loops $\tildegamma$ and $\tildesigma$ which are lifts of geodesics, the critical points of the phase of $\langle \tildesigma(t), \tildegamma(s) \rangle$ will occur at the points of intersection of $\gamma$ and $\sigma$. But, a critical point is a purely local question, and only depends on the tangent vectors $\tildesigma'(t)$ and $\tildegamma'(s)$. And, every tangent vector can be exponentiated to a geodesic. We conclude that for {\em any} Bohr-Sommerfeld loops such that $\gamma$ and $\sigma$ intersect transversely, the critical points of the phase of $\langle \tildesigma(t), \tildegamma(s) \rangle$ will occur at the points of intersection of $\gamma$ and $\sigma$.

Now that we have computed the critical points, let us turn our attention to their stationary phase contributions. In the case of the warm-up example, these were
\[
\sqrt{2} e^{ik\beta/2} \frac{e^{i(\beta/2 - \pi/4)}}{\sqrt{\sin \beta}} \quad \text{and} \quad \sqrt{2} e^{-ik \beta/2} \frac{e^{-i(\beta/2 - \pi/4)}}{\sqrt{\sin \beta}}
\]
respectively. At each point $x$, the first phase factor is clearly just $\omega_x^k$, where $\omega_x \in U(1)$ is defined by $\tildegamma_x = \omega_x \tildesigma_x$. The sign in the argument in the second phase factor can be expressed in a geometric way as the {\em orientation} $\orr_x \in \{+1, -1\}$ of the angle from $\gamma'$ to $\sigma'$ at $x$ ($+1$ if rotating $\gamma'$ to $\sigma'$ agrees with the standard counterclockwise orientation of $S^2$, and $-1$ otherwise). In other words, the contribution at each critical point can be expressed in terms of the local geometry as
\[
     \sqrt{2} \omega_x^k \frac{e^{i \orr_x \left(\theta/2 - \pi/4\right)}}{\sqrt{\sin \theta_x}} \,.
\]
By the same symmetry argument as before, this must hold for {\em general} Bohr-Sommerfeld curves, which proves the theorem.
\end{proof}
\begin{remark} Comparing our formula to \cite[Theorem 4.4b]{borthwick1995legendrian}, there are two main differences. Firstly, in \cite{borthwick1995legendrian}, the angle is allowed to be negative and the square root in the denominator can sometimes therefore be imaginary, whereas our $\theta_x$ and hence our denominator is always positive real. We have isolated all the `sign' information into the $\orr_x$ factor. Also, our formula is the complex conjugate of that in \cite{borthwick1995legendrian}, since our inner product is conjugate linear in the first factor.
\end{remark}
\begin{remark}
Note that the inner product of coherent loop states is $O(1)$ in terms of $k$, as opposed to the inner product of {\em normalized} coherent loop states which is $O(\frac{1}{\sqrt{k}})$, as in Proposition \ref{warmupprop}. 
\end{remark}

We can simplify Theorem \ref{mainthm1} in the case when the loops intersect in precisely two points, having equal angles of intersection at each intersection point, as in Figure \ref{symloops}.

\begin{figure}
\centering
\includegraphics[width=0.4\textwidth]{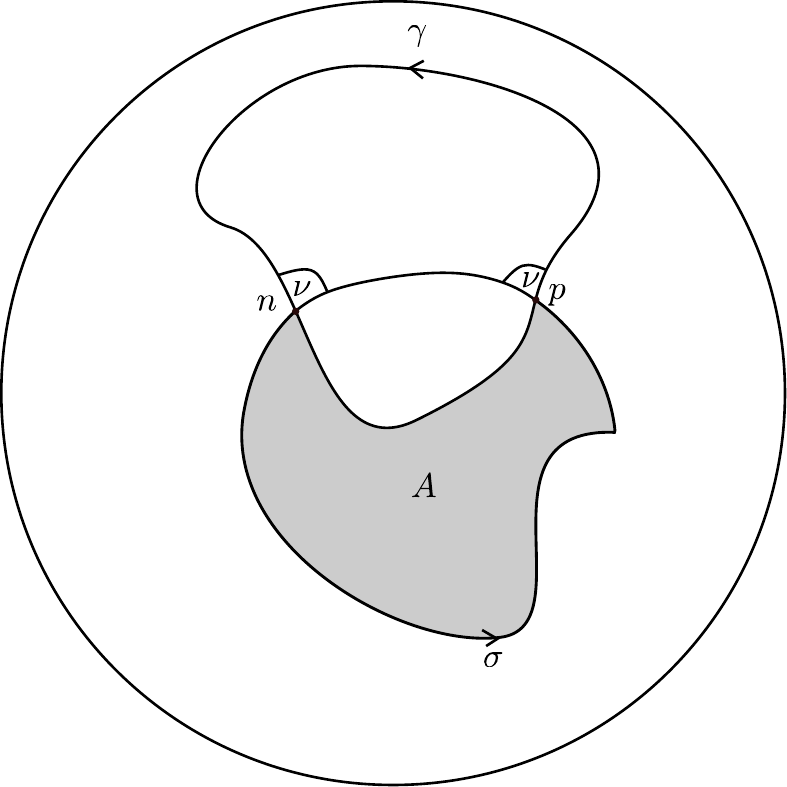}

\caption{\label{symloops} Two oriented loops $\gamma$ and $\sigma$ intersecting twice on $S^2$ with equal angle $\nu$ at the intersection points. The area $A$ is the region to the right of $\gamma$ and to the left of $\sigma$. The two intersection points are distinguished by the orientation of their angles --- at $p$, $\orr_p = +1$ while at $n$, $\orr_n = -1$.  }
\end{figure}

\begin{corollary} \label{corsymcurves} Let $\gamma$ and $\sigma$ be oriented Bohr-Sommerfeld loops on $S^2$ which intersect twice transversely, having equal angle of intersection $\nu$ at each intersection point. Then as $k \rightarrow \infty$ through joint Bohr-Sommerfeld values,
	$$
   \langle \Psi^{(k)}_\tildegamma, \Psi^{(k)}_\tildesigma \rangle \sim \sqrt{\frac{8}{\sin \nu}}\cos(kA/4 + \nu/2 - \pi/4)
	$$
where:
\begin{itemize}
	\item $A$ is the spherical area of the shaded region in Figure \ref{symloops} (the region to the right of $\gamma$ and to the left of $\sigma$),
	\item the relative phase between the lifts $\tildegamma$ and $\tildesigma$ is fixed by the convention that at the positively oriented point $p$, $\tildegamma_p = e^{iA/4} \tildesigma_p$.
\end{itemize}
\end{corollary}

\begin{proof}
From Theorem \ref{mainthm1}, we have
\[
 \langle \Psi^{(k)}_\tildegamma, \Psi^{(k)}_\tildesigma \rangle \sim \sqrt{\frac{2}{\sin \nu}}\left( \omega_p^k e^{i(\nu/2 - \pi/4)} + \omega_n^k e^{-i(\nu/2 - \pi/4)} \right)\, .
\]
From our convention on the relative phase, $\omega_p = e^{iA/4}$. On the other hand, from the holonomy formula in Lemma \ref{hollemma} we can write
\begin{align*}
  \omega_n &= e^{-iA/2} \omega_p \\
   		   &= e^{-iA/4}\,.
\end{align*}
The result follows.
\end{proof}

\subsection{Proof of Littlejohn and Yu's formula}
\begin{figure}
	\centering
    \includegraphics[height=0.28\textheight]{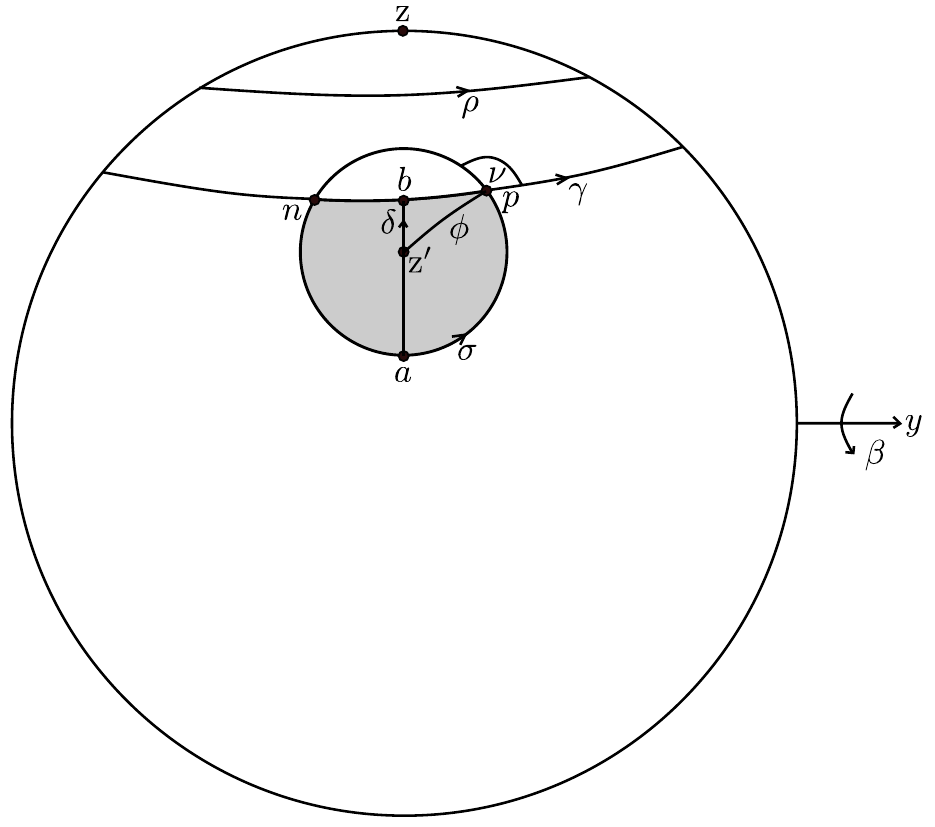}
    \caption{\label{wignermatfig}Acting with $R_y(\beta)$ rotates the loop $\rho$ to $\sigma = R_y(\beta) \rho$. The lunar area $A$ is the area to the right of $\gamma$ and to the left of $\sigma$.}
\end{figure}

Consider a general Wigner small-d matrix element,
\[
 d^j_{m_2 m_1} (\beta) = \langle jm_2 | U_y(\beta) | jm_1 \rangle .
\]
Let us, for brevity, write $\tildegamma \equiv \tildegamma_{m_2}$ and $\tilde{\rho} \equiv \tildegamma_{m_1}$ for the standard lifts of the loops on $S^2$ at constant heights $\cos \theta_1 = m_1/j$ and $\cos \theta_2 = m_2/j$ respectively, as in Definition \ref{standardloop}. We saw in Section \ref{loopsareangs} that 
\[
 | j m_1 \rangle = c \Psi^{k}_{\tilde{\rho}} \quad \text{and } |j m_2 \rangle = c' \Psi^{(k)}_{\tildegamma}
\]
for some positive real constants $c,c'$. Applying $U_y(\beta)$ to the coherent loop state $\Psi^{(k)}_{\tilde{\rho}}$ has a simple geometric formula, as we saw in Section \ref{groupactionloopstates}:
\begin{equation}
  U_y(\beta) \Psi^{(k)}_{\tilde{\rho}} = \int_0^T \psi^{(k)}_{U_y(\beta) \tilde{\rho}(t)} dt
\end{equation}
Now, acting with $SU(2)$ on $\tilde{\rho}(t)$ will produce a parallel-transported lift $\tildesigma(t)$ in $S^3$ of the rotated base curve $\sigma(t)=R_y(\beta)\rho(t)$ in $S^2$.
Therefore, the Wigner matrix element computes as the normalized inner product of coherent loop states, as in the warm-up example from Proposition \ref{warmupprop}:
\begin{equation} \label{wignermatcohloop}
 d^j_{m_2 m_1} (\beta)  = \frac{\langle \Psi^{(k)}_{\tildegamma}, \Psi^{(k)}_{\tildesigma} \rangle}{\sqrt{\langle \Psi^{(k)}_{\tildegamma}, \Psi^{(k)}_{\tildegamma} \rangle \langle \Psi^{(k)}_\tildesigma, \Psi^{(k)}_\tildesigma \rangle}}.
\end{equation}
We say that $\beta$ is {\em classically allowed} if $\gamma$ intersects $\sigma$ transversely, as in Figure \ref{wignermatfig}.

\begin{theorem} \label{lyutheorem}For fixed classically allowed $\beta$, as $j \rightarrow \infty$, 
\[
 d^j_{m' m}(\beta) \sim \sqrt{\frac{2}{j \pi V}} \cos \left( jA/2 + \nu/2 - \pi/4\right)
\]
where $A$ is the area to the right of $\gamma$ and to the left of $\sigma$, $\nu$ is the angle between these curves at their points of intersection, and $V$ is the volume of the parallelopiped spanned by $\z=(0,0,1)$, $\z' = R_y(\beta)(\z)$, and the point of intersection $p$.
\end{theorem}

\begin{figure}
\adjustbox{valign=m}{\includegraphics[width=0.35\textwidth]{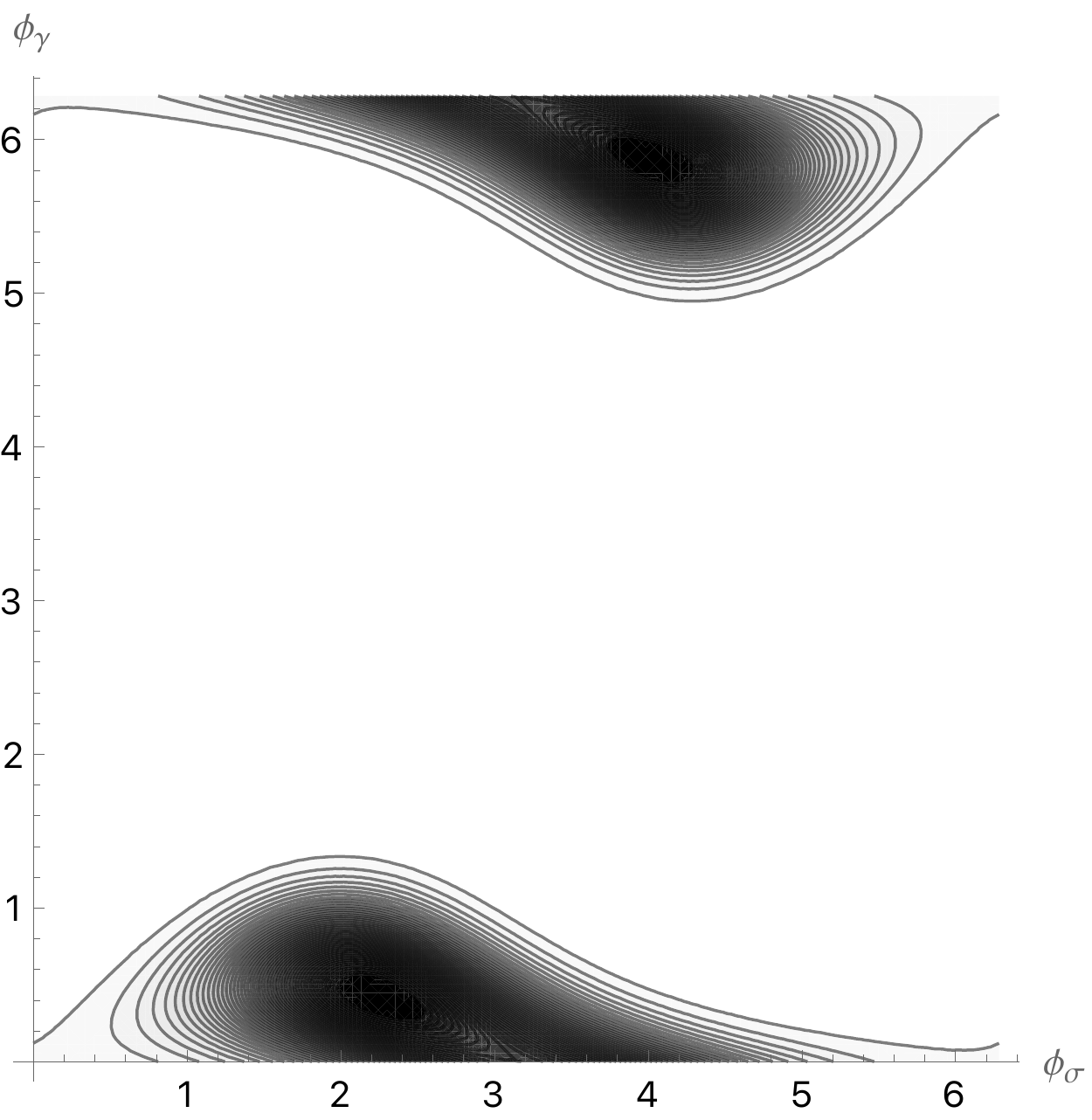}}
\,
\adjustbox{valign=m}{\includegraphics[width=0.05\textwidth]{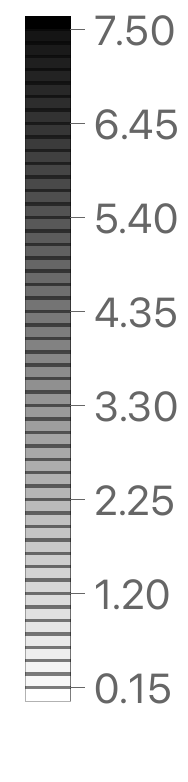}}
\hfill
\adjustbox{valign=m}{\includegraphics[width=0.35\textwidth]{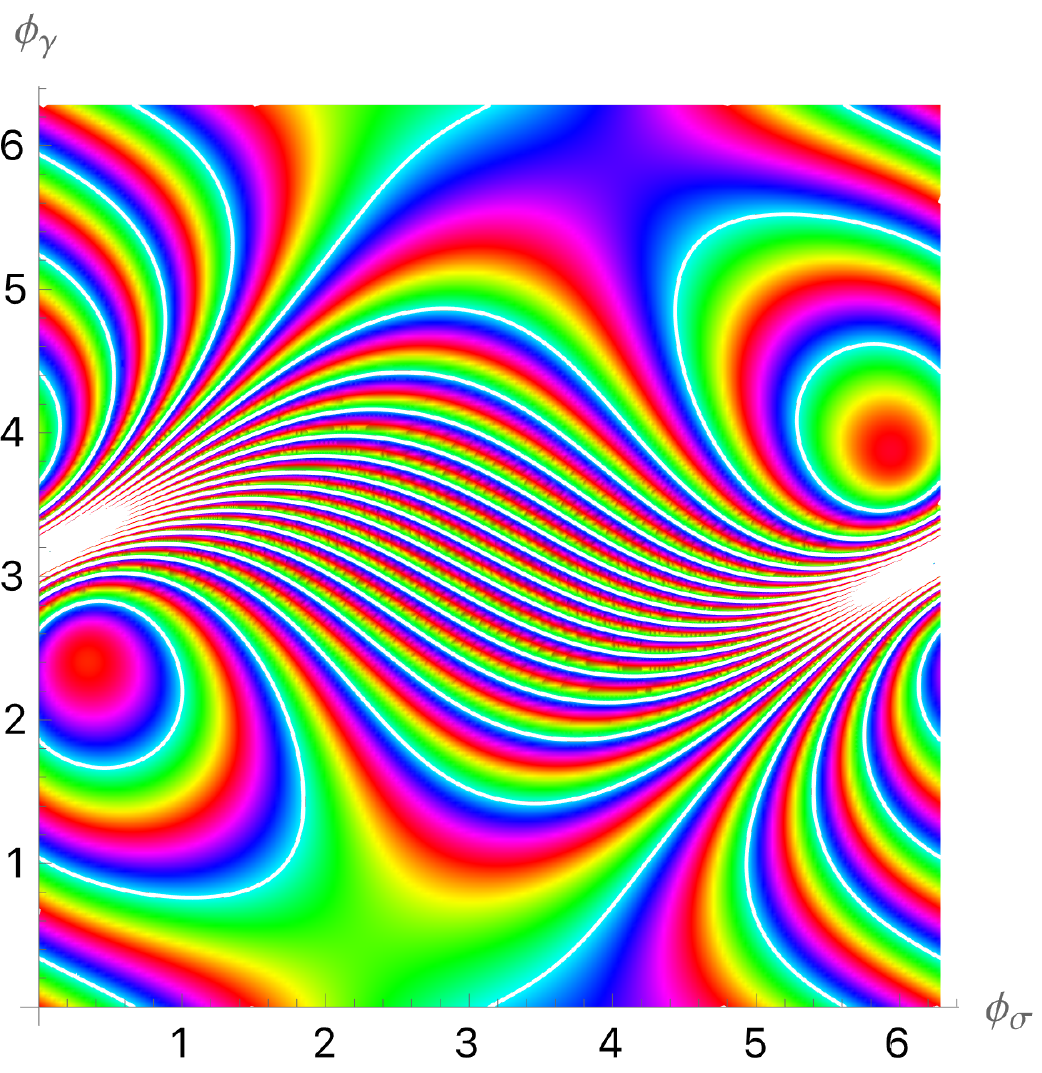}}
\,
\adjustbox{valign=m}{\includegraphics[width=0.15\textwidth]{phasewheel.pdf}}
\caption{\label{rectfig}A continuation of Figures \ref{torusvssphere} and \ref{wignercomp}. The magnitude and phase of $\langle \psi^{(k)}_{\gamma(\phi_\gamma)}, \psi^{(k)}_{\sigma(\phi_\sigma)}\rangle$ are shown, in the case $\beta = 1.4$. Note that the phase diagram has 4 critical points, but only the 2 saddle points at approximately $\pm (2.24, 0.43)$ are simultaneously critical points of the magnitude and hence contribute to the complex stationary phase approximation.}
\end{figure}

\begin{proof}
We can calculate the numerator of \eqref{wignermatcohloop} using Corollary \ref{corsymcurves}, but first we need to check that the standard lifts $\tildegamma_{m'}$ and $\tildegamma_{m}$ we have chosen satisfy
\begin{equation} \label{thingtocheck}
  \tildegamma_{p} = e^{iA/4} \tilde\sigma_{p} \, .
\end{equation}
Now we can identify $e^{-iA/4} = \Hol(p \xleftarrow{\sigma} a \xleftarrow{\delta^{-1}} b \xleftarrow{\gamma^{-1}} p)$, where $\delta$ is the geodesic arc along the line of longitude at $\phi=0$ as shown in Figure \ref{wignermatfig}. So checking \eqref{thingtocheck} is equivalent to checking that $\tildegamma_b$ is equal to the parallel transport of $\tildesigma_a$ along $\delta$ to $b$. This is indeed true, since along the line of longitude $\phi=0$, parallel transporting $u(\theta, 0)$ by $\Delta \theta$ in the $\theta$ direction, which is the same as acting with $U_y(\Delta \theta)$, simply changes $u(\theta, 0)$ to $u(\theta + \Delta \theta, 0)$. Therefore, by Corollary \ref{corsymcurves},
\[
\langle \Psi^{(k)}_{\tildegamma}, \Psi^{(k)}_{\tildesigma} \rangle \sim \sqrt{\frac{8}{\sin \nu}} \cos \left(A/4 + \nu/2 - \pi/4\right) \, .
\]
On the other hand, we have
\[
\langle \Psi^{(k)}_{\tildesigma}, \Psi^{(k)}_{\tildegamma} \rangle \sim \sqrt{\frac{2k}{\pi}} \sin\theta_1, \quad \langle \Psi^{(k)}_{\tildesigma}, \Psi^{(k)}_{\tildegamma} \rangle \sim \sqrt{\frac{2k}{\pi}} \sin\theta_2
\]
from Lemma \ref{norm_lemma}. The result follows from using the law of sines to observe that $V = \sin \beta \sin \theta \sin \phi = \sin \theta \sin \theta' \sin \nu$ (see \cite{littlejohn2009uniform}).
\end{proof}

\begin{figure}
\centering
\adjustbox{valign=m}{\includegraphics[width=0.45\textwidth]{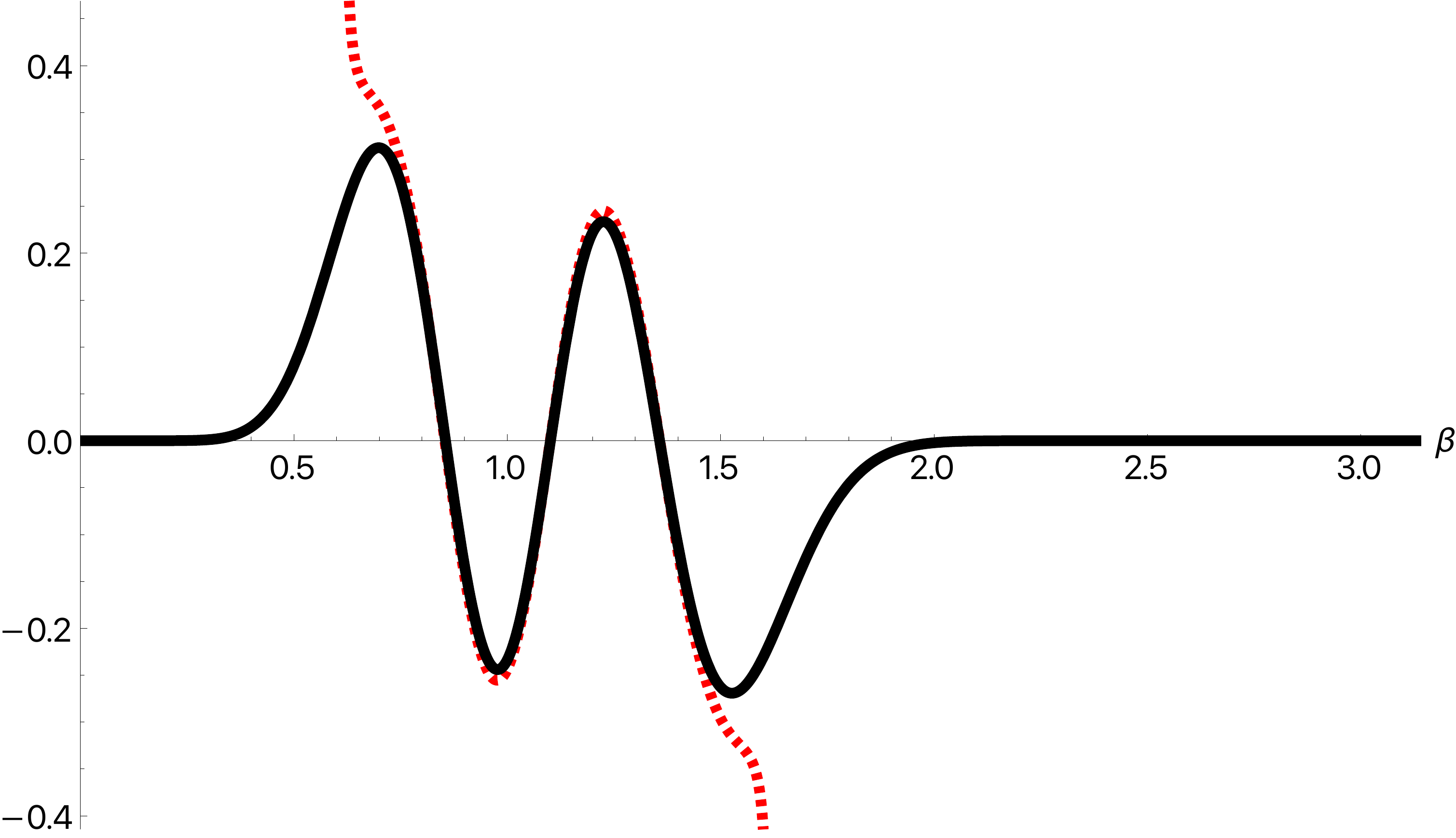}}
\adjustbox{valign=m}{\includegraphics[width=0.45\textwidth]{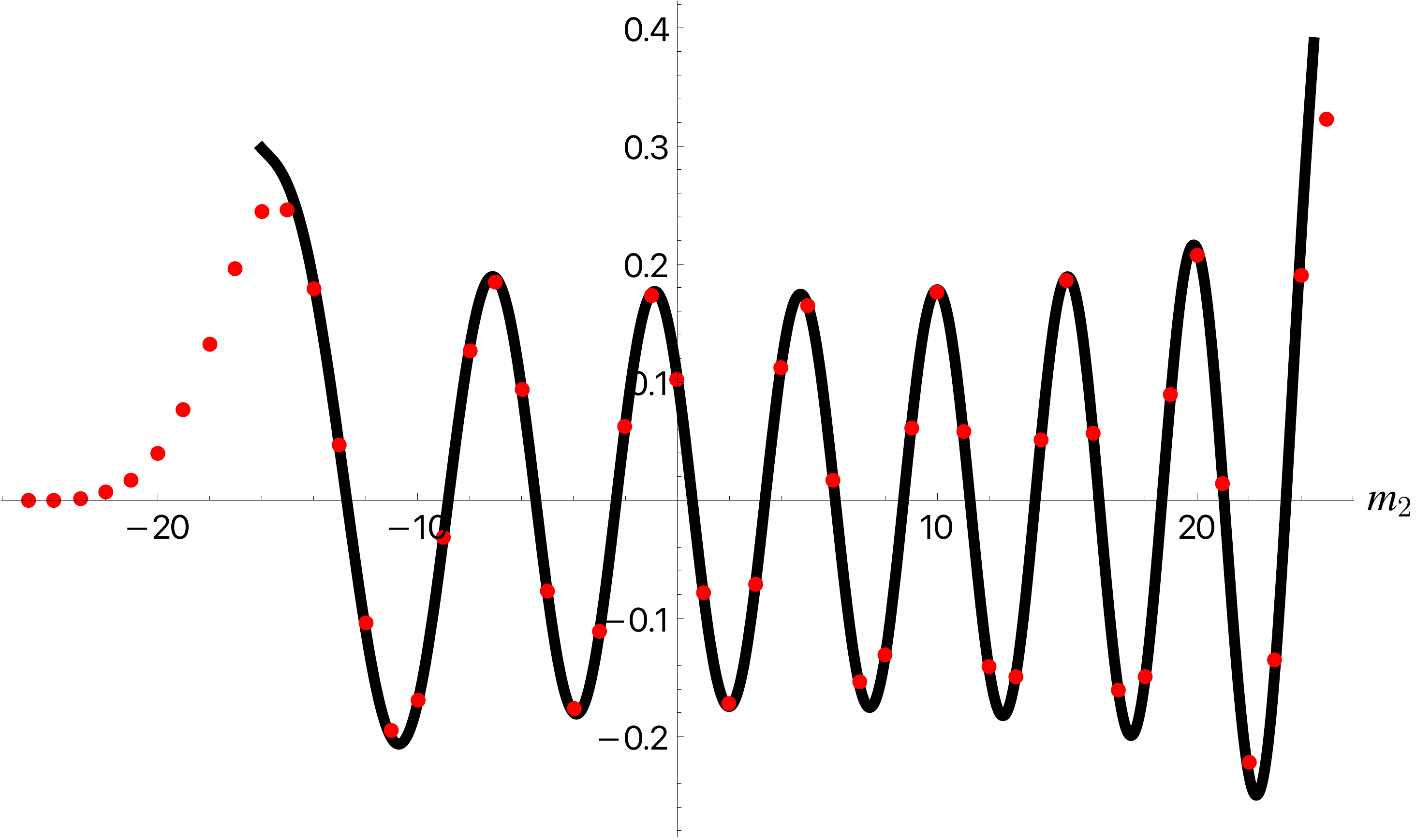}}

\caption{\label{wignercomp} The Wigner matrix element $d^j_{m_2 m_1} (\beta)$ (solid line) versus the `spherical area' approximation given in Theorem \ref{lyutheorem} (dashed line), for $j=25$ and $m_1 = 11$. In (a), $m_2 = 22$ and $\beta$ varies, while in (b), $\beta=1.2$ and $m_2$ varies from $-j$ to $j$. Note that the approximation is only valid in the classically allowed region where $\sigma$ intersects $\gamma$ and breaks down at the boundary of this region. }
\end{figure}

\begin{remark} Littlejohn and Yu had a slightly different (but equivalent) formula \cite[Equation 59]{littlejohn2009uniform}:
\begin{equation} \label{lyformula}
d^j_{mm'} (\beta) \sim \frac{(-1)^{j-m'}}{\sqrt{\pi/2J V}} \cos(\Phi - \pi/4)
\end{equation}
In their setup, the angular momentum sphere had radius $J=j+1/2$, and the Bohr-Sommerfeld loops occurred at the discrete angles given by $\cos\theta = m/J$, $\cos \theta' = m'/J$, which is slightly different to our setup. Moreover, for them, $\Phi$ was the dual area, namely to the {\em left} of $\gamma$ and to the {\em right} of $\sigma$, which leads to the $(-1)^{j-m'}$ term in \eqref{lyformula}. 
\end{remark}

\bibliography{references}

\end{document}